\documentclass[12pt,journal,onecolumn,draftcls]{IEEEtran}

\usepackage{epsfig}
\usepackage{times}
\usepackage{float}
\usepackage{afterpage}
\usepackage{amsmath}
\usepackage{amstext}
\usepackage{amssymb,bm}
\usepackage{latexsym}
\usepackage{color}
\usepackage{graphicx}
\usepackage{amsmath}
\usepackage{amsthm}
\usepackage{graphicx}
\usepackage[center]{caption}
\usepackage{pstricks}
\usepackage{caption}
\usepackage{subcaption}
\usepackage{booktabs}
\usepackage{multicol}
\usepackage{lipsum}

\usepackage[normalem]{ulem}

\newtheorem{prop}{Proposition}

\newtheorem{definition}{Definition}
\newtheorem{theorem}{Theorem}

\newcommand{\eu}{\mathrm{e}}

\newcommand{\Cov}{\mathrm{Cov}}

\newcommand{\snr}{\mathsf{snr}}

\newcommand{\mmse}{\mathrm{mmse}}

\newcommand{\E}{\mathbb{E}}
\newcommand{\Trc}{\mathrm{Tr}}
\newcommand{\cov}{\mathbf{Cov}}
\newcommand{\X}{{\bf X}}
\newcommand{\Z}{{\bf Z}}
\newcommand{\Y}{{\bf Y}}
\newcommand{\I}{{\bf I}}
\newcommand{\supp}{{\mathsf{supp}}}
\newcommand{\gap}{{\mathsf{gap}}}
\newcommand{\pam}{{\mathsf{PAM}}}

\newcommand{\J}{{\bf J}}

\title{On Communication through a Gaussian  Channel with an MMSE Disturbance Constraint}
\author{
\IEEEauthorblockN{Alex Dytso,  Ronit Bustin, Daniela Tuninetti,  Natasha Devroye, H.Vincent Poor, Shlomo Shamai (Shitz) \\
\thanks{Alex Dytso, Daniela Tuninetti and Natasha Devroye are with the department of Electrical and Computer Engineering, University of Illinois at Chicago, IL, Chicago 60607, USA (e-mail: odytso2, danielat, devroye @ uic.edu).  
Ronit Busitn is with  the department of Electrical Engineering - Systems, Tel Aviv University, Tel Aviv  6997801, Israel (email:ronitbustin@post.tau.ac.il). 
H.Vincent Poor is with the department of Electrical and Computer Engineering, Princeton University, NJ, Princeton 08544, USA (email:poor@princeton.edu). 
S. Shamai (Shitz) is with the Department of Electrical Engineering,
Technion-Israel Institute of Technology, Technion City, Haifa  3200003,
Israel (e-mail:  sshlomo@ee.technion.ac.il). 
The work of Alex Dytso, Daniela Tuninetti and Natasha Devroye was partially funded by NSF under award 1422511. 
The work of Ronit Bustin was supported in part by the Women Postdoctoral Scholarship of Israel's Council for Higher Education (VATAT) 2014-2015. 
The work of H. Vincent Poor and Ronit Bustin was partially supported by NSF under awards CCF-1420575 and ECCS-1343210. 
The work of Shlomo Shamai was supported by the Israel Science Foundation
and the S. and N. Grand Research Fund. 
The contents of this article are solely the responsibility of the authors and do not necessarily represent the official views of the funding agencies. 
The work was presented in part in \cite{ITA2016MMSEbounds}. 
} 
 }
}

\begin{document}

\maketitle

\begin{abstract}
This paper considers a Gaussian channel with one transmitter and two receivers. The goal is to maximize the communication rate at the intended/primary receiver subject to a disturbance constraint at the unintended/secondary receiver. The disturbance is measured in terms of minimum mean square error (MMSE) of the interference that the transmission to the primary receiver inflicts on the secondary receiver.

The paper presents a new upper bound for the problem of maximizing the mutual information subject to an MMSE constraint. The new bound holds for vector inputs of any length and recovers a previously known limiting (when the length of vector input tends to infinity) expression from the work of Bustin {\it et al.} 
The key technical novelty is a new upper bound on the MMSE. This  bound allows one to bound the MMSE for all signal-to-noise ratio (SNR) values {\it below} a certain SNR at which the MMSE is known (which corresponds to the disturbance constraint). This  bound complements the `single-crossing point property'  of the MMSE that upper bounds the MMSE for all SNR values {\it above} a certain value at which the MMSE value is known.  The  MMSE upper bound provides a refined characterization of the phase-transition phenomenon which manifests, in the limit as the length of the vector input goes to infinity, as a discontinuity of the MMSE for the problem at hand.

 For vector inputs of size $n=1$, a matching lower bound, to within an additive gap of order $O \left( \log \log \frac{1}{\sf MMSE} \right)$ (where ${\sf MMSE}$ is the disturbance constraint), is shown by means of the mixed inputs technique recently introduced by Dytso {\it et al.}
\end{abstract}

\section{Introduction}

Consider a Gaussian noise channel  with one transmitter and two receivers:
\begin{subequations}
\begin{align}
\Y&=\sqrt{\snr} \ \X +\Z, \label{eq: y at snr}\\
\Y_{\snr_0}&= \sqrt{\snr_0} \ \X+\Z_0,
\end{align}
\label{eq: channel}
\end{subequations} 
where  $\Z,\Z_0,\X,\Y,\Y_{\snr_0} \in \mathbb{R}^n$, $ \Z,\Z_0~\sim \mathcal{N}({\bf 0}, \I)$, and $\X$ and $(\Z,\Z_0)$ are  independent.\footnote{Since there is no cooperation between receivers the capacity depends on $p_{\Y_1,\Y_2|\X}$ only thorough the marginals $p_{\Y_1|\X}$ and $p_{\Y_2| \X}$.} When it will be necessary to stress the SNR at $\Y$ in \eqref{eq: y at snr} we will denote it by $\Y_{\snr}$.  

We denote the  mutual information between the input $\X$ and output $\Y$ as 
\begin{align}
I(\X;\Y) = I(\X,\snr)
:= \E \left[ \log \left( \frac{ p_{\Y|\X}(\Y|\X)}{p_{\Y}(\Y)} \right )\right].
\end{align}
We also denote the mutual information normalized by $n$ as
\begin{align}
 I_n( \X,\snr) :=\frac{1}{n} I(\X,\snr). \label{eq:normalized MI by n}
\end{align}

We  denote the minimum mean squared error (MMSE) of estimating $\X$ from $\Y$ as 
\begin{align}
\mmse(\X|\Y)=\mmse(\X,\snr)
:= \frac{1}{n} \Trc \left(\E \left [ \cov(\X| \Y)\right] \right), \label{eq: def of MMSE}
\end{align}
where $\cov(\X| \Y)$ is the conditional covariance matrix of $\X$ given $\Y$ and is defined as
\begin{align*}
\cov(\X| \Y) :=\E \left [ \left(\X -\E[\X| \Y]\right) \left(\X -\E[\X| \Y]\right)^{\text{T}} | \Y \right].
\end{align*}
Moreover, since the distribution of the noise is fixed, the quantities $I(\X;\Y)$ and $\mmse(\X|\Y)$ are completely determined by $\X$ and $\snr$, and there is no ambiguity in using the notation $I(\X,\snr)$ and $ \mmse(\X,\snr)$.

We consider a scenario in which a message, encoded as $\X$, must be decoded at the primary receiver 
$\Y_\snr$ while it is also seen at the unintended/secondary receiver 
for which it is an interferer.  
This scenario is motivated by the two-user Gaussian Interference Channel  (G-IC), whose capacity is  known only for some special cases. The following  strategies are commonly used  to manage interference in the G-IC:
\begin{enumerate}
\item {\it Interference is  treated as Gaussian noise}: in this approach the interference structure is neglected. It has been shown to be sum-capacity optimal in the so called very-weak interference regime~\cite{ICsumCapacityKramer,motahari2009capacity}, and \cite{annapureddy2009gaussian}. 
\item {\it Partial interference cancellation}: by using the Han-Kobayashi (HK) achievable scheme~\cite{H+K},  part of the interfering message is  jointly  decoded  with  part of the desired signal. Then the decoded part of the interference is subtracted from the received signal, and the remaining part of the desired signal is decoded while  the remaining part of the interference is  treated as  Gaussian noise.  This approach has been shown to be capacity achieving in the strong interference regime~\cite{sato_strong} and optimal within 1/2 bit per channel per user otherwise~\cite{etkin_tse_wang}.  
\item {\it Soft-decoding / estimation}:  the unintended receiver employs soft-decoding of part of the interference. This is enabled by using non-Gaussian inputs and designing the decoders that treat interference as noise by taking into account the correct (non-Gaussian) distribution of the interference. Such scenarios were considered in~\cite{inPraiseOfBadCodes,MoshksarJournal} and \cite{DytsoCodebookJournal}, and shown to be optimal to within either a constant or a $O(\log \log(\snr))$ gap in~\cite{DytsoTINsubmitted}. 
\end{enumerate}

In this paper we look at a somewhat simplified scenario compared to the G-IC as shown in Fig.~\ref{fig:ChannelModel}.
We assume that there is only one message for the primary receiver, and 
the primary user inflicts interference (disturbance) on a secondary receiver. 
The primary transmitter wishes to maximize its comunication rate, while subject to a constraint on the disturbance it inflicts on the secondary receiver.  
The disturbance is 
measured in terms of MMSE.
Intuitively, the MMSE disturbance constraint quantifies the remaining interference after partial interference cancellation or soft-decoding have been performed~\cite{BustinMMSEbadCodes,ShamaiShannonLecture}.
Formally, we aim to solve the following problem.
\begin{subequations}
\begin{definition} {(\emph{max-I problem}.)} For some $\beta \in [0,1]$
\label{def:prob max I st mmse}
\begin{align}
&\mathcal{C}_n(\snr,\snr_0,\beta) := \sup_{\X} I_n(\X,\snr), 
\\&\text{s.t. }  \frac{1}{n}\Trc \left( \E[ \X \X^{\text{T}}] \right) \le 1,\label{eq:power constraint}  \text{ power constraint},
\\&\text{and  }  \mmse(\X,\snr_0) \le \frac{\beta}{1+\beta \snr_0}, \text{ MMSE constraint}.
\label{eq:MMSE constr}
\end{align} 

\end{definition}
\label{eq: mmse constrained cap} 
\end{subequations}
The subscript $n$ in $\mathcal{C}_n(\snr,\snr_0,\beta)$ emphasizes that we seek to find bounds that hold for any input length $n$.
Even though this model is somewhat simplified, compared to the G-IC, 
it can serve as an important building block towards characterizing the capacity of  the G-IC~\cite{BustinMMSEbadCodes, ShamaiShannonLecture}. 

 In~\cite{BustinMMSEbadCodes} the capacity of the channel in Fig.~\ref{fig:ChannelModel} was properly defined and it was shown to be equal to $\lim_{n\to\infty} \mathcal{C}_n(\snr,\snr_0,\beta)$.  Note that $\mathcal{C}_n(\snr,\snr_0,\beta)$ does not denote the capacity since the MMSE does not `single-letterize.'    Finally,  in~\cite[Sec. VI.3]{ShamaiShannonLecture} and \cite[Sec. VIII]{BustinMMSEbadCodes} it was conjectured that the optimal input for $\mathcal{C}_1(\snr,\snr_0,\beta)$ is discrete.

\begin{figure}
        \centering
                \includegraphics[width=9cm]{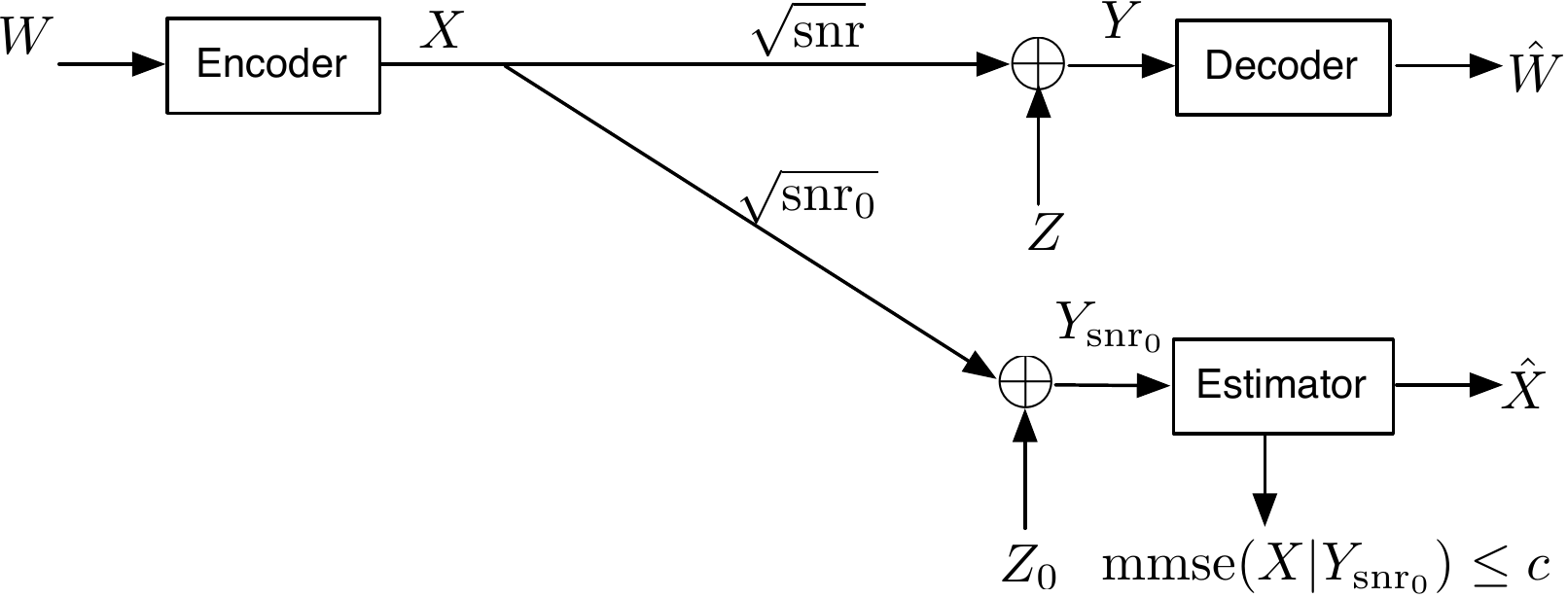}
                \caption{Channel Model.}
                \label{fig:ChannelModel}
\end{figure}

\subsection{Notation}
Throughout the paper we adopt the following notational conventions: deterministic scalar quantities are denoted by lowercase letters and deterministic vector quantities are denoted by lowercase bold letters; matrices are denoted by bold uppercase letters; random variables  are denoted by uppercase letters and random vectors are denoted by bold uppercase letters; all logarithms are taken to be base $\eu$; we denote the support of a random variable $A$ by $\supp(A)$; $X \sim \pam\left(N \right)$ denotes the pulse-amplitude modulation (PAM) constellation, i.e.,  the uniform probability mass function over a zero-mean constellation with $|\supp(X)|=N$ points, minimum distance $d_{\min(X)}$, and therefore average energy $\mathbb{E}[X^2] = d_{\min  \left(X \right)}^2\frac{N^2-1}{12}$;  ordering notation ${ \bf A} \succeq {\bf B}$ implies that ${\bf A}-{\bf B}$ is a positive semidefinite matrix; 
we denote the Fisher information matrix of the random vector ${\bf A}$ by $\J({\bf A})$;   for $x \in \mathbb{R}$  we let $[x]^{+} := \max(x, 0)$ and $\log^+(x) := [\log(x)]^+$; we use the Landau notation $f(x) = O(g(x))$
to mean that for some $c>0$ there exists an $x_0$ such that $f(x) \le cg(x)$ for all $x \ge x_0$.

\subsection{On Presentation of Results} 
Throughout the paper we will plot normalized quantities, where the normalization is with respect to the same quantity when the input is  $\mathcal{N}({\bf 0},\I)$. 
For example, for mutual information $I_n(\X,\snr)$ in \eqref{eq:normalized MI by n} we will plot 
\begin{align}
d(\X,\snr)&:=\frac{I_n(\X,\snr)}{\frac{1}{2}\log(1+\snr)},\label{eq:degree of freedom def}
\end{align}
while for MMSE in \eqref{eq: def of MMSE} we will plot 
\begin{align}
D(\X,\snr)&:=\frac{\mmse(\X,\snr)}{\frac{1}{1+\snr}}=(1+\snr) \cdot \mmse(\X,\snr). \label{eq: MMSE dim def}
\end{align}

 In particular,  at high $\snr$ the  quantity in \eqref{eq:degree of freedom def} is commonly referred to as the {\it degrees of freedom}  \cite{Jafar:2008:alignment}   and the quantity in \eqref{eq: MMSE dim def} as  the {\it MMSE dimension} \cite{mmseDim}. Moreover, it is well known that under the block-power constraint in \eqref{eq:power constraint}, a Gaussian input maximizes both the mutual information and the MMSE \cite{Cover:InfoTheory}, and thus the quantities $d(\X,\snr), \ D(\X,\snr)$ have a natural meaning of multiplicative loss of the inputs $\X$ compared to  the Gaussian input. Fig.~\ref{fig1:PresentingResults} compares normalized and unnormalized quantities. 
\begin{figure*}
        \centering
 \begin{subfigure}[t]{0.5\textwidth}
                \includegraphics[width=8cm]{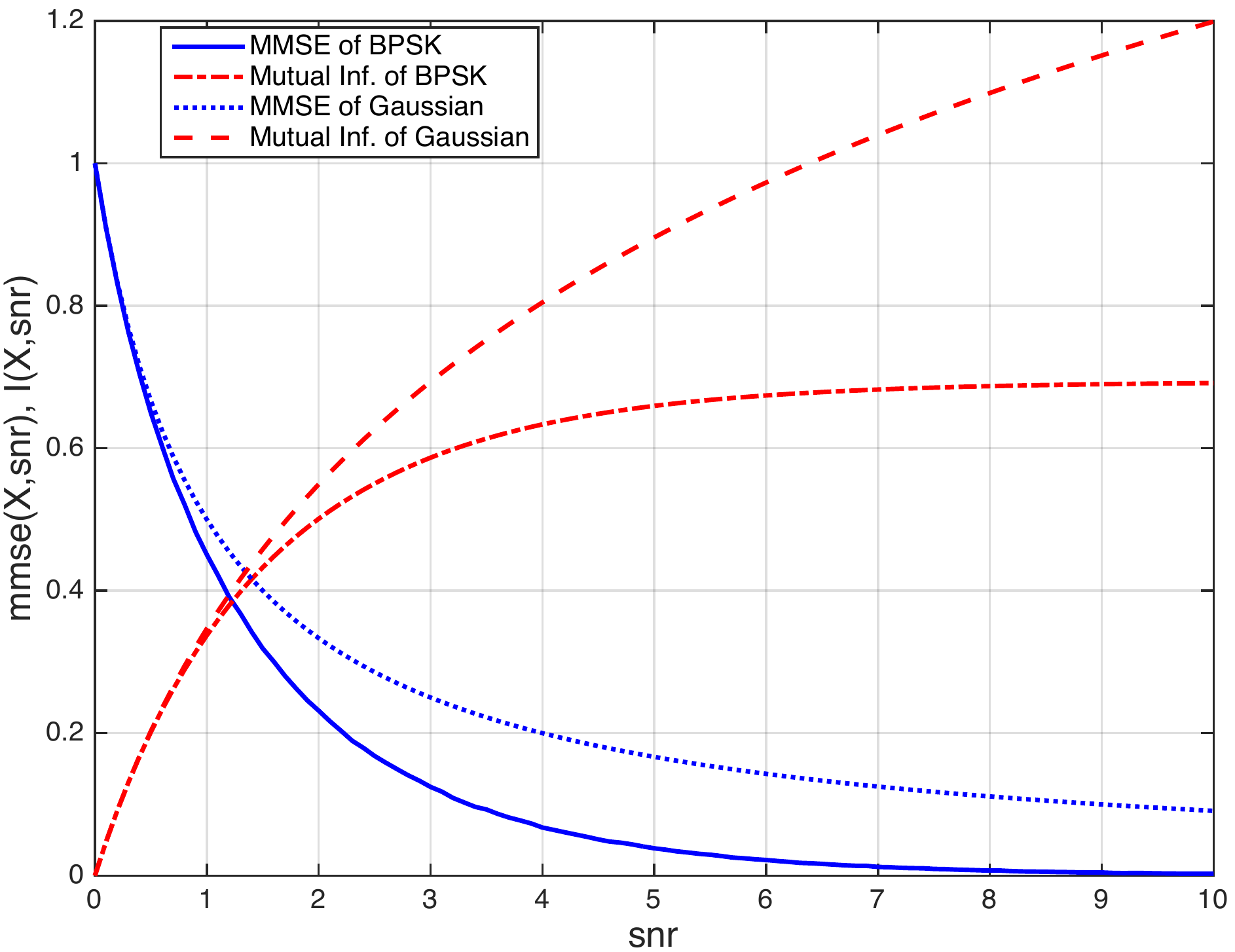}
                \caption{ Unnormalized plot.}
                \label{fig:ExampleNotNormalized}
        \end{subfigure}%
~ 
           \begin{subfigure}[t]{0.5\textwidth}
                \includegraphics[width=8cm]{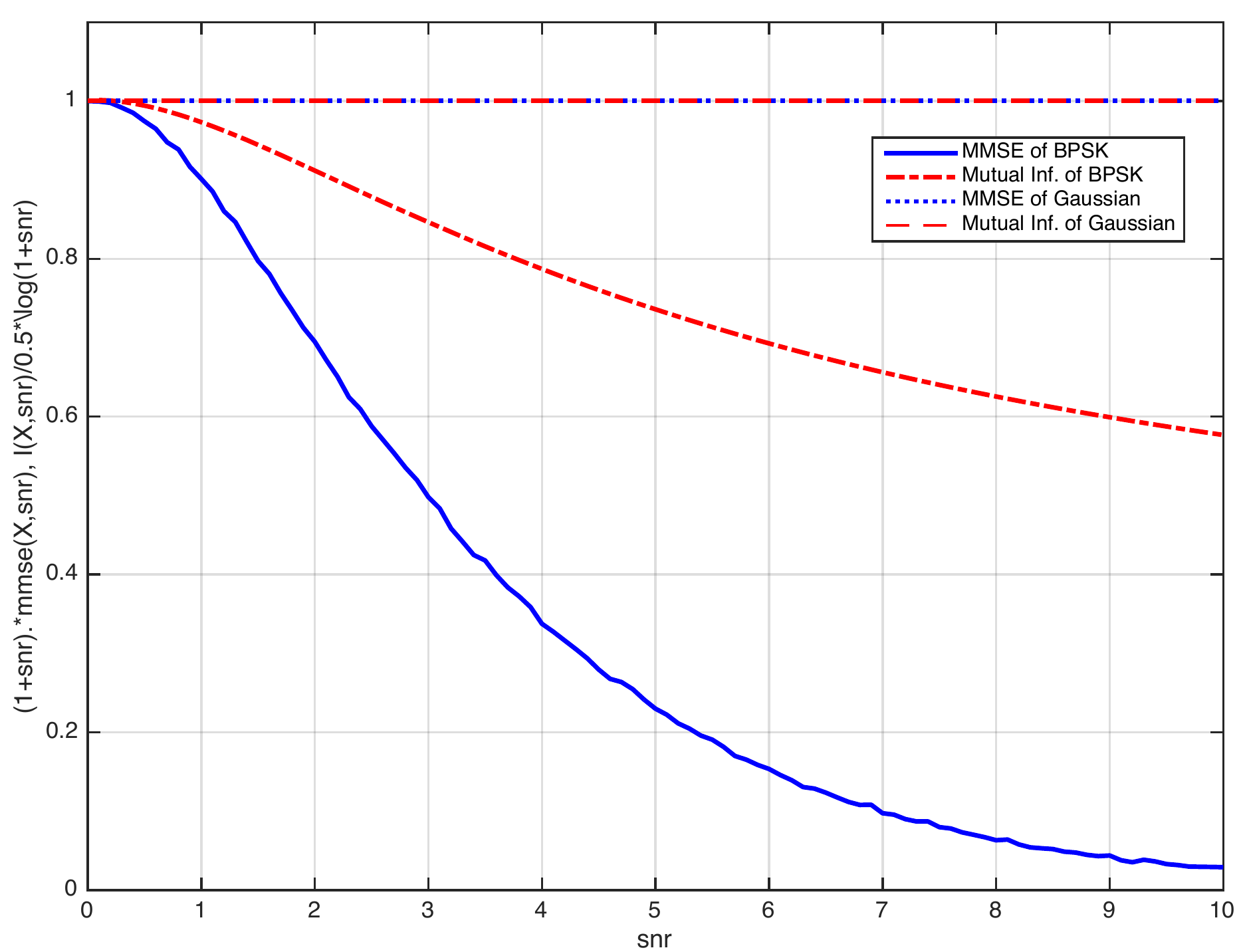}
                \caption{Normalized plot (dashed red and dotted blue lines overlap).}
                \label{fig:ExampleNormalized}
        \end{subfigure}
        \caption{Comparing mutual informations and MMSE's for BPSK and Gaussian inputs.  Fig.~\ref{fig:ExampleNormalized} clearly shows the multiplicative loss  of BPSK, for both mutual information and MMSE, compared to a Gaussian input.}
        \label{fig1:PresentingResults}
\end{figure*}

\section{Past Work and Paper Contributions}
\label{sec:past}
The mutual information and the MMSE   are  related, for any  input $\X$, via the so called {\it I-MMSE relationship}~\cite[Theorem 1]{I-MMSE}. %
\begin{prop}
\label{prop:IMMSE} (\emph{I-MMSE relationship} \cite{I-MMSE}.)
The I-MMSE relationship is given by the derivative relationship
\begin{subequations}
\begin{align}
\frac{d}{d \snr} I_n(\X,\snr) =\frac{1}{2} \mmse(\X,\snr),
\end{align}
or the integral relationship~\cite[Eq.(47)]{I-MMSE} 
\begin{align}
I_n(\X,\snr)= \frac{1}{2} \int_{0}^{\snr} \mmse(\X, t) dt.
\end{align}
\label{eq:I-MMSE}
\end{subequations}
\end{prop}
In order to develop bounds on $\mathcal{C}_n(\snr,\snr_0,\beta)$ we require bounds on the MMSE.  
An important bound on the MMSE is the following {\it linear MMSE (LMMSE) upper bound}.
\begin{prop}
\label{prop:LMMSE bounds} (\emph{LMMSE bound \cite{I-MMSE}}.)
For any  $\X$ and $\snr>0$ it holds that
\begin{subequations}
\begin{align}
\mmse(\X,\snr) \le \frac{1}{\snr}. \label{eq: LMMSE decorrolator} 
\end{align}
If  $ \frac{1}{n} \Trc \left( \E[ \X \X^{\text{T}}] \right)  \le \sigma^2$, then for any $\snr \ge 0$ 
\begin{align}
\mmse(\X,\snr) \le \frac{\sigma^2}{1+\sigma^2 \snr},  \label{eq: LMMSE power}
\end{align}
\end{subequations}
where equality in~\eqref{eq: LMMSE power}  is achieved  iff $\X \sim \mathcal{N}(0, \sigma^2 \I)$.
\end{prop}

Another important bound for the MMSE is the {\it single-crossing point property} (SCPP) bound developed in~\cite{GuoMMSEprop} for $n=1$ and extended in~\cite{BustinMMSEparallelVectorChannel} to any $n\ge 1$.
\begin{prop}
\label{prop:SCP} 
 \emph{(SCPP \cite{BustinMMSEparallelVectorChannel}.)} For any fixed $\X$, suppose that $\mmse(\X,\snr_0)=\frac{\beta}{1+\beta \snr_0}$, for some fixed $\beta \ge 0$. Then for all $\snr \in [\snr_0, \infty)$ we have that
 \begin{subequations}
\begin{align}
\mmse(\X,\snr) \le \frac{\beta}{1+\beta \snr},  \label{eq: single-cross upper}
\end{align}
and  for all $\snr \in [0,\snr_0)$
\begin{align}
\mmse(\X,\snr) \ge \frac{\beta}{1+\beta \snr}.  \label{eq: single-cross lower}
\end{align}
\end{subequations}
\end{prop}

In words, Proposition~\ref{prop:SCP} means that if we know that the value of MMSE at $\snr_0$ is given by  $\mmse(\X,\snr)=\frac{\beta}{1+\beta \snr_0}$ then for all higher SNR values ($\snr_0 \le \snr$)  we have  the upper bound in~\eqref{eq: single-cross upper} and for all lower  SNR values $(\snr \le \snr_0)$ we have the lower bound  in~\eqref{eq: single-cross lower}.
Unfortunately, Proposition~\ref{prop:SCP} does not provide an upper bound on $\mmse(\X,\snr)$ for  $\snr \in [0 ,\snr_0)$ and one of the goals of this paper is to fill  this gap.  Note that upper bounds on the MMSE are useful, thanks to the I-MMSE relationship, as tools to derive converse results, and have been used in \cite{guo2013interplay, GuoMMSEprop, BustinMMSEparallelVectorChannel}, and \cite{MMSEforBCandWireTap} to name a few.

Motivated by the search  for the complementary upper bound to the SCPP we define the following problem.
\begin{subequations} 
\begin{definition}  {(\emph{max-MMSE problem}.)} For some $\beta \in [0,1]$
\label{def:prob max mmse st mmse}
\begin{align}
&\mathrm{M}_n(\snr,\snr_0,\beta) := \sup_{\X}
\mmse(\X,\snr),
\\& \text{ s.t. }\frac{1}{n} \Trc \left( \E[ \X \X^{\text{T}}] \right)  \le 1, \label{eq:power constr: for MMSE}
\\& \text{ and } \mmse(\X,\snr_0) \le \frac{\beta}{1+\beta \snr_0}.
\label{eq:MMSE constr: for MMSE}
\end{align}
\end{definition}
\label{eq: mmse constrained MMSE} 
\end{subequations}

Clearly, $\mathrm{M}_n(\snr,\snr_0,\beta)  \le \mathrm{M}_{\infty}(\snr,\snr_0,\beta)$ for all finite $n$.
Observe that the max-MMSE problem  in~\eqref{eq: mmse constrained MMSE}  and the max-I problem in~\eqref{eq: mmse constrained cap}  have  different objective functions but have the same  constraints. 
This is also a good place to point out that neither of the max-MMSE and max-I problems   falls under the category of convex optimization. This follows from the fact that the MMSE is a strictly concave function in the input distribution \cite{FunctionalPropMMSE}. Therefore, the set of input distributions, defined by   \eqref{eq:power constr: for MMSE} and \eqref{eq:MMSE constr: for MMSE},  over which we are optimizing, might not be convex.

Note that Proposition~\ref{prop:SCP} gives a solution to the max-MMSE problem in~\eqref{eq: mmse constrained MMSE} for $\snr \ge \snr_0$ and any $n \ge 1$ as follows:
\begin{align}
\mathrm{M}_n(\snr,\snr_0,\beta)=\frac{\beta}{1+\beta \snr}, \text{ for } \snr \ge \snr_0,
\end{align}
achieved by $\X \sim \mathcal{N}(0 , \beta \I)$. 
Therefore in the rest of the paper the treatment of the max-MMSE problem will focus only  on the regime  $\snr \le \snr_0$.

The case  $ n \to\infty$ of the max-MMSE problem in~\eqref{eq: mmse constrained MMSE} was solved for   random codes using statistical physics in ~\cite[Section V-C]{MerhavStatisticalPhysics}  and  generalized in ~\cite[Theorem 2]{BustinMMSEbadCodes} as follows:
\begin{align}
\mathrm{M}_{\infty}(\snr,\snr_0,\beta)=\left \{ \begin{array}{cc}  \frac{1}{1+\snr}, & \snr< \snr_0,\\
 \frac{\beta}{1+ \beta \snr}, & \snr\ge \snr_0,
 \end{array} \right.,
 \label{eq:outer bound for inf MMSE}
\end{align}
achieved by using superposition coding with Gaussian codebooks.   For other recent links between random codes, the MMSE and statistical physics see \cite{huleihel2014analysis}.

Clearly there is a discontinuity in~\eqref{eq:outer bound for inf MMSE}  at $\snr=\snr_0$ for $\beta<1$.  This fact is a well known property of the MMSE, and it is referred to as a \emph{phase transition}~\cite{MerhavStatisticalPhysics}. It is also well known that, for any finite $n$,  $\mmse(\X,\snr)$ is a continuous function of $\snr$~\cite{GuoMMSEprop}. Putting these two facts together we have that, for any finite $n$,  the objective function $\mathrm{M}_{n}(\snr,\snr_0,\beta)$ must be continuous in $\snr$ and converge to a function with a jump-discontinuity at $\snr_0$ as $n \to \infty$. Therefore, $\mathrm{M}_{n}(\snr,\snr_0,\beta)$ must be of the following form:
\begin{align}
&\mathrm{M}_n(\snr,\snr_0,\beta)  =  \left \{  \begin{array}{ll} 
\frac{1}{1+\snr}, & \snr \le \snr_L, \\
T_n(\snr,\snr_0,\beta), & \snr_L \le  \snr \le \snr_0, \\
\frac{\beta}{1+\beta \snr},& \snr_0 \le \snr,
 \end{array} \right. 
\label{eq: M regions}
\end{align}
for some $\snr_L$. 
In this paper we seek to characterize $\snr_L$ in~\eqref{eq: M regions} and the continuous function $T_n(\snr,\snr_0,\beta)$ such that 
\begin{subequations}
\begin{align}
T_n(\snr_L,\snr_0,\beta) &=\frac{1}{1+\snr_L},\\
T_n(\snr_0,\snr_0,\beta)&= \frac{\beta}{1+\beta \snr_0},
\end{align}
\end{subequations}
and give scaling bounds on the width of the phase transition region defined as 
\begin{align}
W_n:=\snr_0-\snr_L.
\label{eq: width phase transition region}
\end{align}

\begin{figure}
        \centering
                \includegraphics[width=8cm]{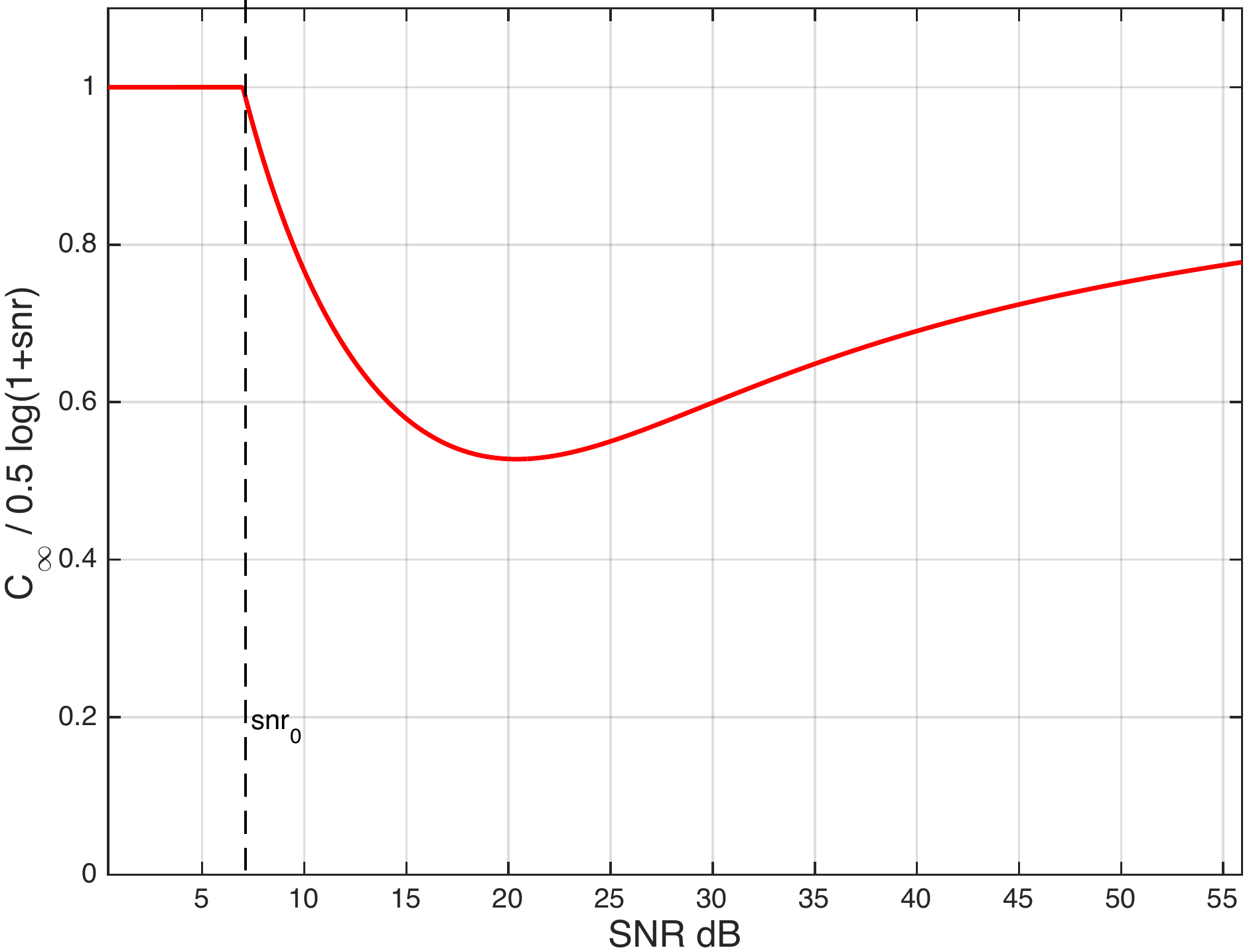}
                \caption{Plot of $ \frac{\mathcal{C}_{\infty}(\snr,\snr_0,\beta)}{\frac{1}{2} \log(1+\snr)} $ vs. $\snr$dB, for $\beta=0.01$, $\snr_0= 5=  6.989$ dB.}
                \label{fig:Cinf}
\end{figure}

Back to the max-I problem in~\eqref{eq: mmse constrained cap}. Clearly $\mathcal{C}_n (\snr,\snr_0,\beta)$  is a non-decreasing function of $n$.
In ~\cite[Theorem. 3]{BustinMMSEbadCodes} it was shown that
\begin{align}
&\mathcal{C}_{\infty}(\snr,\snr_0,\beta)
= \lim_{n \to \infty} \mathcal{C}_{n}(\snr,\snr_0,\beta),
\notag\\
&
= \left \{ \begin{array}{cc}
\frac{1}{2}\log(1+\snr),      &\snr \le \snr_0,\\
\frac{1}{2} \log(1+\beta \snr) +\frac{1}{2} \log \left( 1+\frac{\snr_0(1-\beta)}{1+\beta \snr_0}\right), &\snr \ge \snr_0,
 \end{array} \right. 
\notag\\
& =\frac{1}{2} \log^{+}  \left( \frac{1+\beta \snr}{1+\beta \snr_0} \right)+\frac{1}{2} \log \left(1+\min(\snr,\snr_0) \right),
\label{eq: Outer bound}
\end{align}
which is achieved by using superposition coding with Gaussian codebooks.  
Fig.~\ref{fig:Cinf} shows a plot of  $\mathcal{C}_{\infty}(\snr,\snr_0,\beta)$ normalized by the capacity of the point-to-point channel  $\frac{1}{2} \log(1+\snr)$. 
The region $\snr \le \snr_0$ (flat part of the curve) is where the MMSE constraint is inactive since the channel with $\snr_0$ can decode the interference and guarantee zero MMSE.
The regime $\snr \ge \snr_0$ (curvy part of the curve) is where the receiver with  $\snr_0$ can no-longer decode the interference and the MMSE constraint becomes active, which in practice
is the more interesting regime because the secondary receiver experiences  `weak interference' that can not be fully decoded (recall that in this regime superposition coding appears to be the best achievable strategy for the G-IC, but it is unknown  whether it achieves capacity~\cite{etkin_tse_wang}).

The importance of studying models  of communication systems with disturbance constraints has been recognized previously.  For example, in~\cite{BandermerCommDist} Bandemer {\it et al.}  studied the following problem related to the max-I problem in~\eqref{eq: mmse constrained cap}. 
\begin{definition}  {(\emph{Bandemer {\it et al.} problem}.)} For some $R\ge 0$
\begin{subequations}
\begin{align}
\mathcal{I}_n(\snr,\snr_0,R)&:=\max_{\X}   I_n(\X, \snr),
\\& \text{s.t. } \frac{1}{n} \Trc \left( \E[ \X \X^{\text{T}}] \right)   \le 1, 
\label{eq: Bandemer P constr}
\\& \text{ and } I_n(\X, \snr_0) \le R. 
\label{eq: Bandemer MI constr}
\end{align}
\label{eq: BandemerProblem}
\end{subequations} 
\end{definition} 
In~\cite{BandermerCommDist} it was shown that the optimal solution for $\mathcal{I}_n(\snr,\snr_0,R)$, for any $n$, is attained by  $ \X \sim  \mathcal{N} \left(0, \alpha \I \right)$ where  $\alpha = \min \left(1,\frac{\eu^{2R}-1}{\snr_0} \right)$; here $\alpha$ is such that the  most stringent constraint between~\eqref{eq: Bandemer P constr} and~\eqref{eq: Bandemer MI constr} is satisfied  with equality. In other words, the optimal input is  i.i.d. Gaussian with power reduced such that the disturbance constraint in~\eqref{eq: Bandemer MI constr} is not violated. 

Observe that the max-I problem in~\eqref{eq: mmse constrained cap} and the one in~\eqref{eq: BandemerProblem} have the same objective function but have  different constraints. 
The relationship between the constraints in~\eqref{eq:MMSE constr} and~\eqref{eq: Bandemer MI constr} can be explained as follows. The constraint in~\eqref{eq:MMSE constr} imposes a maximum value on the function $\mmse(\X,\snr)$ at $\snr=\snr_0$, while the constraint in~\eqref{eq: Bandemer MI constr}, via the integral I-MMSE relationship in~\eqref{eq:I-MMSE}, imposes a constraint on the area below the function $\mmse(\X,\snr)$  in the range $\snr \in [0,\snr_0]$.

Measuring the disturbance with the mutual information as in~\eqref{eq: BandemerProblem}, in contrast  to the MMSE as in~\eqref{eq: mmse constrained cap}, suggests that it is always optimal to use Gaussian codebooks with the reduced power without any rate splitting.  Moreover, while the mutual information constraint in \eqref{eq: BandemerProblem} limits the amount of information transmitted to the unintended receiver, it may not be the best choice when one models the interference, since any information that can be reliably decoded is not really  interference. For this reason, it has been argued in~\cite{BustinMMSEbadCodes} and \cite{ShamaiShannonLecture} that the max-I problem in \eqref{eq: mmse constrained cap} with the MMSE disturbance constraint is a more suitable building block to study the G-IC and understand the key role of rate splitting.

\subsection{Contributions and Paper Outline} 
The main contributions of the paper are as follows.
In Section~\ref{sec:main results} we summarize our main results:
\begin{itemize}
 \item Theorem~\ref{thm: diff MMSE}, our main technical result, provides new upper bounds 
 for the max-MMSE problem for arbitrary $n$ that complement the SCPP bound.
 \item Proposition~\ref{prop:snrL:diff MMSE bound} provides a lower bound on the width of the phase transition region of the order of $\frac{1}{n}$.
 \item Proposition~\ref{prop: new upper bound on I with derivative} provides a new upper bound 
 for the max-I problem for arbitrary $n$.
\item Proposition~\ref{prop:gap} shows that, for the case of $n=1$, superposition of discrete and Gaussian inputs, termed \emph{mixed input} inputs in~\cite{DytsoTINsubmitted}, achieves the proposed upper bound on the max-I problem from Proposition~\ref{prop: new upper bound on I with derivative}
to within an additive gap of order $\log \log \frac{1}{\mmse(X,\snr_0)}$.
 \end{itemize}
In Section~\ref{sec: derivative} we develop bounds on the derivative of MMSE, which we use to prove Theorem~\ref{thm: diff MMSE}:
\begin{itemize}
\item Proposition~\ref{prop: bound on derivative}  considerably refines existing bounds on the derivative of MMSE for $n=1$ and generalizes them to any $n$.
\item In Section~\ref{sec: proof of diff of MMSE}, by using Proposition~\ref{prop: bound on derivative},  we prove Theorem~\ref{thm: diff MMSE}.
\end{itemize}
In Section~\ref{sec:power constratint} we explore whether the MMSE constraint  implies a power constraint:
\begin{itemize}
\item  Proposition~\ref{prop: counter example}  demonstrates that there exists an input distribution that can transmit at full power while satisfying any MMSE constraint.
\item 
Proposition~\ref{prop:new bound for small snr0}  develops new bounds on the MMSE under the assumption  that  the derivative of the MMSE exists  at $\snr=0^{+}$.  This assumption is also a necessary and sufficient condition for the MMSE constraint to imply a power constraint.
\end{itemize} 

Most proofs can be found in the Appendix.

\section{Main Results}
\label{sec:main results}

\subsection{max-MMSE problem: upper bounds on $\mathrm{M}_n(\snr,\snr_0,\beta)$} 

We start by giving bounds on the phase transition region of $\mathrm{M}_{n}(\snr,\snr_0,\beta)$ defined in~\eqref{eq: M regions}. The bound in  Theorem~\ref{thm: diff MMSE} is referred to as the D-bound because it was derived through the technique of bounding the derivative of the MMSE. 
\begin{theorem}
\label{thm: diff MMSE} ({\it D-Bound.})
For any $\X$ and $0< \snr \le \snr_0$,  let $\mmse(\X,\snr_0)=\frac{\beta}{1+\beta \snr_0}$ for some $\beta \in [0,1]$. Then 
\begin{subequations}
\begin{align}
 &\mmse(\X,\snr)  \le  \mmse(\X,\snr_0)+k_n \left( \frac{1}{\snr}-\frac{1}{\snr_0} \right)-\Delta, \label{eq:mainBound}
\\&
   k_n \le n+2, \ \Delta 
   =0.
   \label{eq: delta main bound: no power}
\end{align}
If $\X$ is such that   $\frac{1}{n} \Trc \left( \E[ \X \X^{\text{T}}] \right)  \le 1$ then 
\begin{align}
\Delta &:=\Delta_{\eqref{eq:mainBoundWithPower}}=\int_{\snr}^{\snr_0}\frac{1}{\gamma^2(1+\gamma)^2} d\gamma. \notag\\
&=2\log \left( \frac{1+\snr_0}{1+\snr}\right) -2\log \left(\frac{\snr_0}{\snr}\right)+\frac{1}{1+\snr} -\frac{1}{1+\snr_0}+\frac{1}{\snr}-\frac{1}{\snr_0}.   
\label{eq:mainBoundWithPower}
\end{align}
\end{subequations}
\end{theorem}
\begin{IEEEproof}
See Section~\ref{sec: proof of diff of MMSE}.
\end{IEEEproof}
The bound on $\mathrm{M}_n(\snr,\snr_0,\beta)$  in~\eqref{eq:mainBound}  is depicted in Fig.~\ref{fig:Example of Phase}, where:
\begin{itemize}
\item the red solid line is the $\mathrm{M}_\infty(\snr,\snr_0,\beta)$ upper bound on $\mathrm{M}_1(\snr,\snr_0,\beta)$, and
\item the  blue  dashed-dotted line is the new upper bound on $\mathrm{M}_1(\snr,\snr_0,\beta)$ from Theorem~\ref{thm: diff MMSE}. 
\end{itemize}
Observe that the new bound provides a tighter and continuous upper bound on $\mathrm{M}_1(\snr,\snr_0,\beta)$ than the trivial upper bound given by  $\mathrm{M}_\infty(\snr,\snr_0,\beta)$. 

We next show how fast the phase transition region shrinks with $n$ as $n\to\infty$.  
\begin{prop} 
\label{prop:snrL:diff MMSE bound}
The bound in~\eqref{eq:mainBound}, with $\Delta=0$, from Theorem~\ref{thm: diff MMSE} intersects  the LMMSE bound in~\eqref{eq: LMMSE decorrolator} from Proposition~\ref{prop:LMMSE bounds}   at
\begin{subequations}
\begin{align}
\snr_{L}&=\snr_0\frac{1+\beta \snr_0}{\frac{k_n}{k_n-1} +\beta \snr_0}= O \left(  \left(1-\frac{1}{n} \right)\snr_0 \right).\end{align}
 \label{eq: intersection}
Thus, the  width of the phase transition region is given, for $k_n$ in \eqref{eq: delta main bound: no power}, by 
 \begin{align}
W_n&=\frac{1}{k_n-1}\frac{\snr_0}{ \frac{k_n}{k_n-1}+\beta \snr_0}=O \left(\frac{1}{n} \right). \label{eq:width phase trans}
 \end{align}
\end{subequations} 
\end{prop}
\begin{IEEEproof}
See Appendix~\ref{app: snrL for D-bound}. 
\end{IEEEproof}

In Proposition~\ref{prop:snrL:diff MMSE bound} we found the intersection between the LMMSE bound $\frac{1}{\snr}$ in \eqref{eq: LMMSE decorrolator} and the  bound  in~\eqref{eq:mainBound} from Theorem~\ref{thm: diff MMSE}. Unfortunately, for the power constraint case, the intersection of  the LMMSE bound $\frac{1}{1+\snr}$  in \eqref{eq: LMMSE power} and the bound in~\eqref{eq:mainBoundWithPower} cannot be found analytically.  However,  the solution can be computed efficiently by using numerical methods. Moreover, the asymptotic behavior of  the phase transition region is still given by $O \left(\frac{1}{n} \right)$. The bound in Theorem~\ref{thm: diff MMSE} for several values of $n$ is shown in Fig.~\ref{fig:ExampleWith- n}, where:
\begin{itemize}
\item the red line is the $\mathrm{M}_\infty(\snr,\snr_0,\beta)$ bound on $\mathrm{M}_n(\snr,\snr_0,\beta)$, and
\item the blue line is the bound on $\mathrm{M}_n(\snr,\snr_0,\beta)$ from Theorem~\ref{thm: diff MMSE} for $n=1,3,15$ and $70$. 
\end{itemize}
 We observe that the new bound provides a refined characterization of the phase transition phenomenon for finite $n$ and, in particular, it recovers the bound in \eqref{eq:outer bound for inf MMSE} as $n \to \infty$.

\begin{figure*}
        \centering
 \begin{subfigure}[t]{0.5\textwidth}
                \includegraphics[width=8cm]{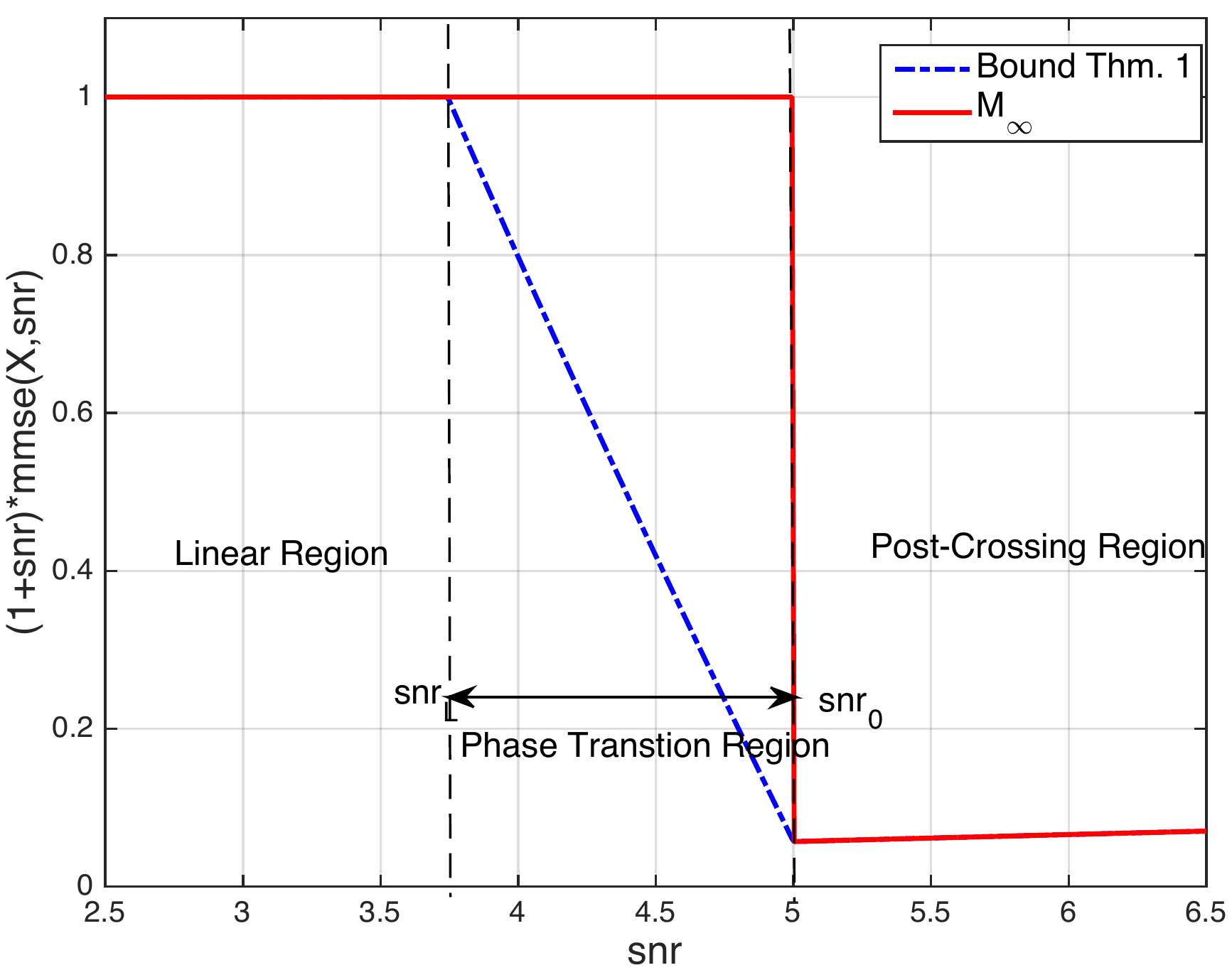}
                \caption{
                For $n=1$, $\snr_0=5$ and $\beta=0.01$.}
                \label{fig:Example of Phase}
        \end{subfigure}%
~ 
           \begin{subfigure}[t]{0.5\textwidth}
                \includegraphics[width=8cm]{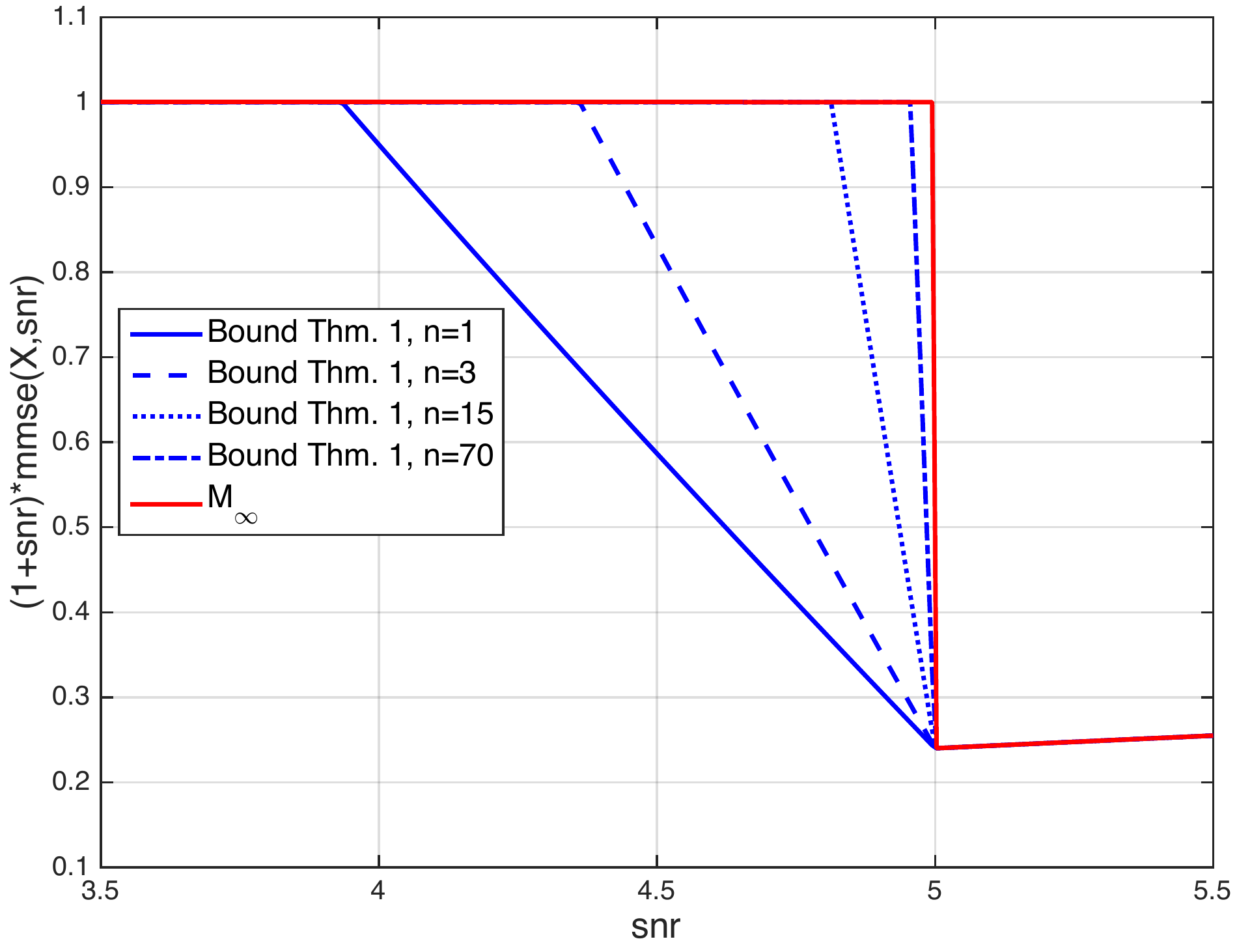}
                \caption{
                For several values of $n$, $\snr_0=5$ and $\beta=0.05$. }
                \label{fig:ExampleWith- n}
        \end{subfigure}
       \caption{Bounds on $\mathrm{M}_n(\snr,\snr_0,\beta)$  vs. $\snr$.}
        \label{fig: Bounds on Mn}
\end{figure*}

\subsection{max-I problem: upper bounds on $\mathcal{C}_n(\snr,\snr_0,\beta)$}
Using the previous novel bound on $\mathrm{M}_n(\snr,\snr_0,\beta)$ in Theorem~\ref{thm: diff MMSE}
we can find new upper bounds on $\mathcal{C}_n(\snr,\snr_0,\beta)$
by integration as follows:
\begin{align}
 &\mathcal{C}_n(\snr,\snr_0,\beta) \le \frac{1}{2} \int_{0}^\snr  \mathrm{M}_n(t,\snr_0,\beta) d t \notag\\
 &=\frac{1}{2} \log(1+\snr_L)+\frac{1}{2} \int_{\snr_L}^{\snr_0}  T_n(t, \snr_0, \beta) d t +\frac{1}{2} \log \left( \frac{1+\beta \snr}{1+\beta \snr_0}\right), \text{ for } \snr_0 \le \snr, \label{eq: form of the upper bound in weak interference}
 \end{align}
and 
 \begin{align}
  &\mathcal{C}_n(\snr,\snr_0,\beta) 
  \le \frac{1}{2} \int_{0}^\snr  \mathrm{M}_n(t,\snr_0,\beta) d t \notag\\
  &\le  \frac{1}{2} \log(1+ \min(\snr_L,\snr)) +\frac{1}{2} \int_{\min(\snr_L,\snr)}^{\snr}  T_n(t, \snr_0, \beta) d t, \text{ for } \snr_0 \ge \snr. \label{eq: form of the upper bound in strong interference}
\end{align}
By using Theorem~\ref{thm: diff MMSE}  (with finite power assumption) to bound $T_n(t, \snr_0, \beta)$ we get the following upper bounds on $ \mathcal{C}_n(\snr,\snr_0,\beta)$.
\begin{prop} 
\label{prop: new upper bound on I with derivative}
For any $0 \le \snr_0$, $\beta \in [0,1]$, and $\snr_L$ given in Proposition~\ref{prop:snrL:diff MMSE bound}, we have that  for $\snr_0 \le \snr$
\begin{align}
 \mathcal{C}_n(\snr,\snr_0,\beta)  &\le    \mathcal{C}_\infty(\snr,\snr_0,\beta) -\Delta_{\eqref{eq: Delta: I bound through derivative: weak interf}},
\end{align}
and for $\snr_0 \ge \snr$
\begin{align}
 \mathcal{C}_n(\snr,\snr_0,\beta)  &\le   \mathcal{C}_\infty(\snr,\snr_0,\beta) -\Delta_{\eqref{eq: Delta: I bound through derivative: strong interf}},
\end{align}
where
\begin{align}
 0 &\le \Delta_{\eqref{eq: Delta: I bound through derivative: weak interf}}=  \frac{1}{2} \log\left(\frac{1+\snr_0}{1+\snr_L}\right)-\frac{1}{2}\frac{\beta (\snr_0-\snr_L)}{1+\beta \snr_0} -\frac{(n+2)}{2} \log \left( \frac{\snr_0}{\snr_L} \right)+ \frac{(n+2)(\snr_0-\snr_L)}{2\snr_0}\notag\\
& + \frac{1}{2} \left( (2\snr_L+1) \log \left(\frac{\snr_0(1+\snr_L)}{\snr_L (1+\snr_0)} \right)-\frac{\snr_0-\snr_L}{1+\snr_0} -\frac{\snr_0-\snr_L}{\snr_0} \right)= O\left( \frac{1}{n}\right) \label{eq: Delta: I bound through derivative: weak interf},  
\end{align}
and 
\begin{align}
0 &\le \Delta_{\eqref{eq: Delta: I bound through derivative: strong interf}}=\frac{1}{2}\log \left( \frac{1+\snr} {1+\min(\snr_L,\snr)}\right) -\frac{\beta (\snr-\min(\snr_L,\snr))}{2(1+\beta \snr_0)}
 \notag \\
&-\frac{(n+2)}{2} \log \left( \frac{\snr}{\min(\snr_L,\snr)} \right) +\frac{(n+2)(\snr-\min(\snr_L,\snr))}{2\snr_0}\notag\\
& +  \frac{1}{2} \left(  (2 \min(\snr_L,\snr)+1) \log \left( \frac{1+\min(\snr_L,\snr)}{\min(\snr_L,\snr)}\right)
\right. - (2 \snr+1) \log \left( \frac{1+\snr}{\snr}\right) \notag\\
& + 2(\snr-\min(\snr_L,\snr)) \log \left( \frac{1+\snr_0}{\snr_0} \right) \left.-\frac{\snr-\min(\snr_L,\snr)}{\snr_0} -\frac{\snr-\min(\snr_L,\snr)}{1+\snr_0}\right) \notag\\
&=O\left( \frac{1}{n}\right). \label{eq: Delta: I bound through derivative: strong interf} 
\end{align}
\end{prop}

Fig.~\ref{fig:Assymptoticly} compares the bounds on $\mathcal{C}_n(\snr,\snr_0,\beta)$ in \eqref{eq: Outer bound} from Proposition~\ref{prop: new upper bound on I with derivative} with $\mathcal{C}_\infty(\snr,\snr_0,\beta)$ for several values of $n$. The figure shows how the new bounds in Proposition~\ref{prop: new upper bound on I with derivative} improve on the trivial $\mathcal{C}_\infty(\snr,\snr_0,\beta)$ bound for finite $n$.
 
\begin{figure}
        \centering
                \includegraphics[width=8cm]{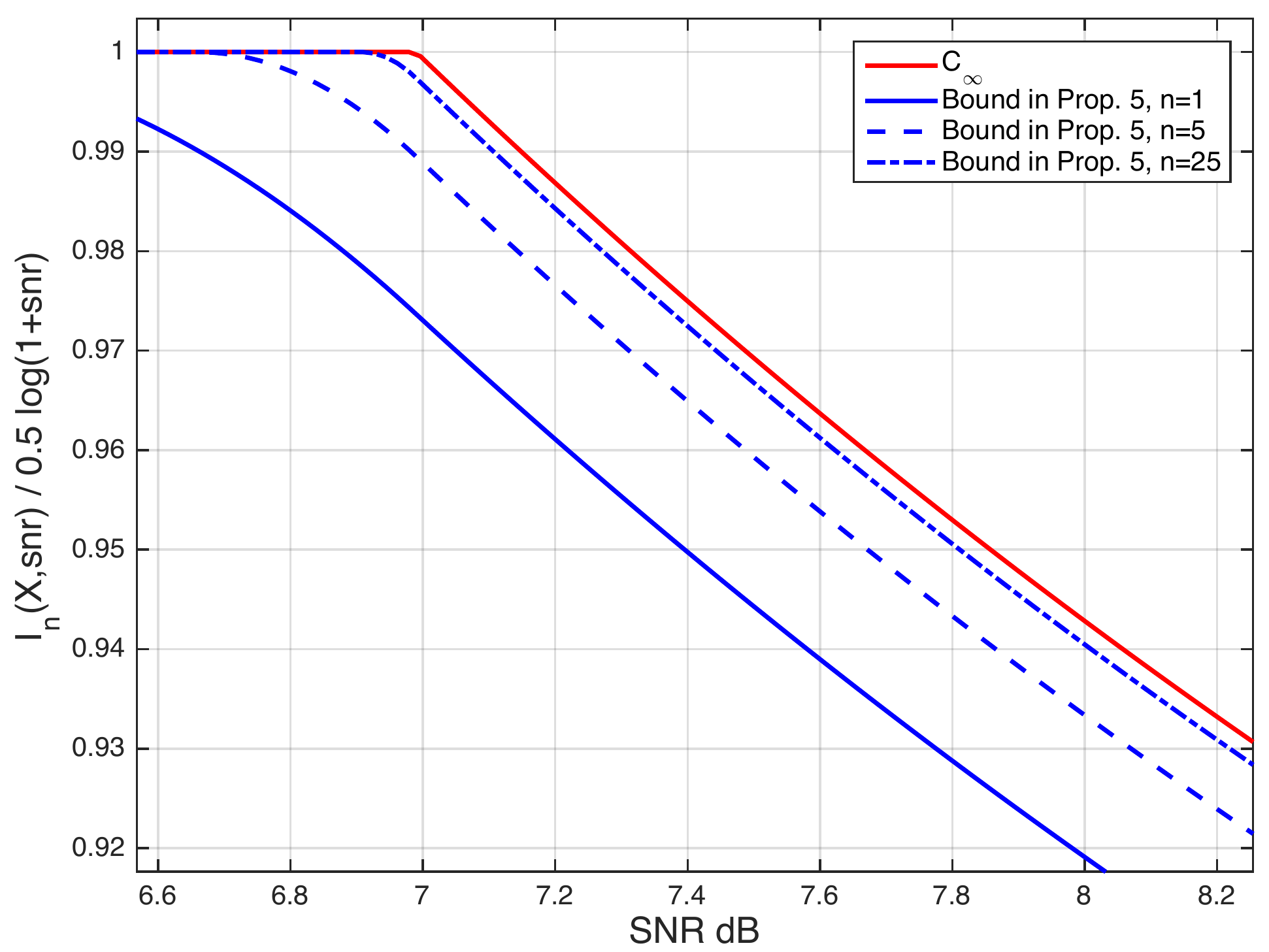}
                \caption{ Bounds on  $\mathcal{C}_n(\snr,\snr_0,\beta)$ vs. $\snr$, for $\beta=0.1$ and $\snr_0=5= 6.9897 $ dB. }
                \label{fig:Assymptoticly}
\end{figure}

\subsection{max-MMSE problem: achievability of $\mathrm{M}_{1}(\snr,\snr_0,\beta)$} 
 In this section we propose an input that will be used in the achievable strategy for both 
the max-I problem and the max-MMSE problem with input length $n=1$. 
This input is referred to as \emph{mixed input}~\cite{DytsoTINsubmitted}
and is defined as
\begin{align}
X_\mathrm{mix}:=\sqrt{1-\delta}X_D+\sqrt{\delta}X_G, \ \delta \in [0,1], \label{eq:input}
\end{align}
where  $X_G$ and $X_D$ are independent, 
$X_G \sim \mathcal{N}(0,1)$, 
$\E[X_D^2] \le 1$, and where the distribution of $X_D$ and the parameter $\delta$ are to be optimized over. 
The  input $X_\mathrm{mix}$ exhibits a  decomposition property via which the MMSE and the mutual information 
can be written as the sum of the MMSE and the mutual information 
of the $X_D$ and $X_G$ components, albeit  at different SNR values.

\begin{prop} 
\label{prop: decomposition}
For $X_\mathrm{mix}$ defined in~\eqref{eq:input} we have that
\begin{subequations}
\begin{align}
  I(X_\mathrm{mix},\snr)&= I\left(X_D, \frac{\snr(1-\delta)}{1+\delta \snr}\right)+I(X_G, \snr \ \delta ),
\\\mmse(X_\mathrm{mix},\snr)
  &=\frac{1-\delta }{(1+\snr \delta)^2}\mmse \left(X_D,\frac{\snr(1-\delta)}{1+\delta \snr}\right)
+  \delta \ \mmse(X_G, \snr \ \delta ).
\end{align}
\label{eq: decompositions of Mixed}
\end{subequations}
\end{prop}
\begin{IEEEproof}
See Appendix~\ref{app: proof prop decomposition}.
\end{IEEEproof}

 Observe that Proposition~\ref{prop: decomposition}  implies that, in order for mixed inputs (with $\delta<1$) to comply with the MMSE constraint in \eqref{eq:MMSE constr} and \eqref{eq:MMSE constr: for MMSE}, the MMSE of   $X_D$  must satisfy
\begin{align}
\mmse \left(X_D,\frac{\snr_0(1-\delta)}{1+\delta \snr_0}\right)\le \frac{(\beta -\delta) (1+\delta \snr_0) }{(1-\delta) (1+\beta \snr_0)}.\label{eq: constraint on discrete}
\end{align}
The bound in \eqref{eq: constraint on discrete} will be helpful in choosing  the parameter $\delta$ later on.  

 When $X_D$ is a discrete random variable with $\supp(X_D)=N$ we use the following bounds from~\cite[App. C]{OptimalPowerAllocationOfParChan} and~\cite[Rem. 2]{DytsoTINsubmitted}.
\begin{prop} (\cite{OptimalPowerAllocationOfParChan,DytsoTINsubmitted})
\label{prop:bounds on I and mmse of discrete}
For a discrete random variable $X_D$ such that $p_i=\Pr(X_D=x_i)$, for $i \in [1:N]$, we have that
\begin{subequations}
\begin{align}
\mmse(X_D,\snr) &\le d_{\max}^2  \sum_{i=1}^N p_i \eu^{-\frac{\snr}{8} d_{i}^2}
\label{eq: MMSE discrete} 
,\\
I(X_D,\snr) &\ge H(X_D) -\frac{1}{2} \log \left( \frac{\pi }{6}\right)-\frac{1}{2}\log \left( 1+\frac{12}{d^2_{\min}} \mmse(X_D,\snr)\right),
\label{eq: MI discrete} 
\end{align}
where  
\begin{align}
d_{\ell}&:=\min_{x_i \in \supp(X_D):  i \neq \ell } |x_{\ell}-x_i|,\\
d_{\min}&:=\min_{ \ell \in[1:N]}  d_{\ell},\\
d_{\max}&:=\max_{ x_k, x_i \in \supp(X_D)}|x_k-x_i|.
\end{align}
\end{subequations} 
\end{prop}

Proposition~\ref{prop: decomposition} and Proposition~\ref{prop:bounds on I and mmse of discrete} are particularly useful  because they will allow us to design Gaussian and discrete components of the mixed input independently. 

 Fig.~\ref{fig:Mmixed} shows  upper and lower bounds on $\mathrm{M}_{1}(\snr,\snr_0,\beta)$  where we show the following:
\begin{itemize}
\item 
The   $\mathrm{M}_{\infty}(\snr,\snr_0,\beta)$  upper bound   in \eqref{eq:outer bound for inf MMSE} (solid red line) ;  
\item 
The upper bound from Theorem~\ref{thm: diff MMSE}  with finite power (dashed cyan line); 
\item 
The Gaussian-only input lower bound (green line), with $X \sim \mathcal{N}(0,\beta) $, where the power has been reduced to meet the MMSE constraint;
\item 
The mixed input lower bound (blue dashed line), with the input in~\eqref{eq:input}. We used Proposition~\ref{prop: decomposition} where we optimized over $X_D$ for $\delta=\beta \frac{\snr_0}{1+ \snr_0}$.  The choice of $\delta$ is motivated by the scaling property of the MMSE,  that is, $\delta \mmse(X_G, \snr \delta)=\mmse(\sqrt{\delta}X_G, \snr )$, and the constraint on the discrete component in \eqref{eq: constraint on discrete}.  That is, we chose $\delta$ such that the power of $X_G$ is approximately $\beta$  while the MMSE constraint  on $X_D$ in \eqref{eq: constraint on discrete} is not equal to zero. 
The input $X_D$ used in Fig.~\ref{fig:Mmixed} was found by  a local search algorithm on the space of distributions with $N=3$, and resulted in $X_D=[-1.8412, -1.7386, 0.5594]$ with $P_X=[0.1111, 0.1274, 0.7615]$, which we do not claim to be optimal; 
\item
The discrete-only input lower bound (Discrete 1 brown dashed-dotted line),
with  \newline $X_D= [-1.8412, -1.7386, 0.5594]$ with $P_X=[0.1111, 0.1274, 0.7615]$, that is,
the same discrete part of the above mentioned mixed input. This is done for completeness, and  to compare the performance of the MMSE of the discrete component of the mixed input with and without the Gaussian component; and
\item 
The discrete-only input lower bound (Discrete 2 dotted magenta line), 
with \newline $X_D=[-1.4689,-1.1634, 0.7838]$ with $P_X=[0.1282, 0.2542, 0.6176]$,
which was found by using  a local search algorithm on the space of discrete-only distributions with $N=3$ points. 
\end{itemize} 
The choice of $N=3$ is motivated by the fact that it requires roughly $N=\lfloor \sqrt{1+\snr_0} \rfloor$ points  for the PAM input to approximately achieve capacity of the point-to-point channel with SNR value $\snr_0$. 

On the one hand, Fig.~\ref{fig:Mmixed} shows that, for $\snr \ge \snr_0$, a Gaussian-only input with  power reduced to $\beta$  maximizes $\mathrm{M}_{1}(\snr,\snr_0,\beta)$ in agreement with the SCPP bound (green line). 
On the other hand, for $\snr \le \snr_0$, we see that  discrete-only inputs (brown dashed-dotted line and magenta dotted line) achieve higher MMSE than a Gaussian-only input with reduced power. Interestingly, unlike Gaussian-only inputs, discrete-only inputs do not have to reduce power in  order to meet the MMSE constraint. The reason discrete-only inputs can use full power, as per the power constraint only,
is because their MMSE  decreases fast enough (exponentially in SNR, as seen in~\eqref{eq: MMSE discrete}) to comply with the MMSE constraint.  However, for $\snr \ge \snr_0$, the  behavior of the MMSE of discrete-only inputs, as opposed to mixed inputs, prevents it from being optimal; this is due to their exponential tail behavior in \eqref{eq: MMSE discrete}.  This further motivates determining whether the MMSE constraint can imply a power constraint, which we shall investigate  in Section~\ref{sec:power constratint}. 
The mixed input (blue dashed line)  gets the best of both (Gaussian-only and discrete-only) worlds: 
it  has the behavior of Gaussian-only inputs for $\snr \ge \snr_0$ (without any reduction in power) and  
        the behavior of discrete-only inputs for $\snr \le \snr_0$.
This behavior of  mixed inputs turns out to  be important for the max-I problem, 
where we need to choose an input that has the largest area under the MMSE curve.  

Finally, Fig.~\ref{fig:Mmixed} shows the achievable MMSE with another discrete-only input (Discrete 2, dotted magenta line)
that achieves higher MMSE than the mixed input for $\snr \le \snr_0$
but lower than the mixed input for $\snr \ge \snr_0$. 
This is again due to the tail behavior of the MMSE of discrete inputs.   The reason this second discrete input is not used as a component of the mixed inputs, is because this choice would violate the MMSE constraint on $X_D$ in \eqref{eq: constraint on discrete}. 
 Note that the difference between Discrete 1 and Discrete 2 is that, Discrete 1 was found as an optimal discrete component of  a mixed input (i.e., $\delta=\beta \frac{\snr_0}{1+\snr_0}$), while the Discrete 2 was found as an optimal discrete input without a Gaussian component (i.e., $\delta=0$).

\begin{figure}
        \centering
                \includegraphics[width=8cm]{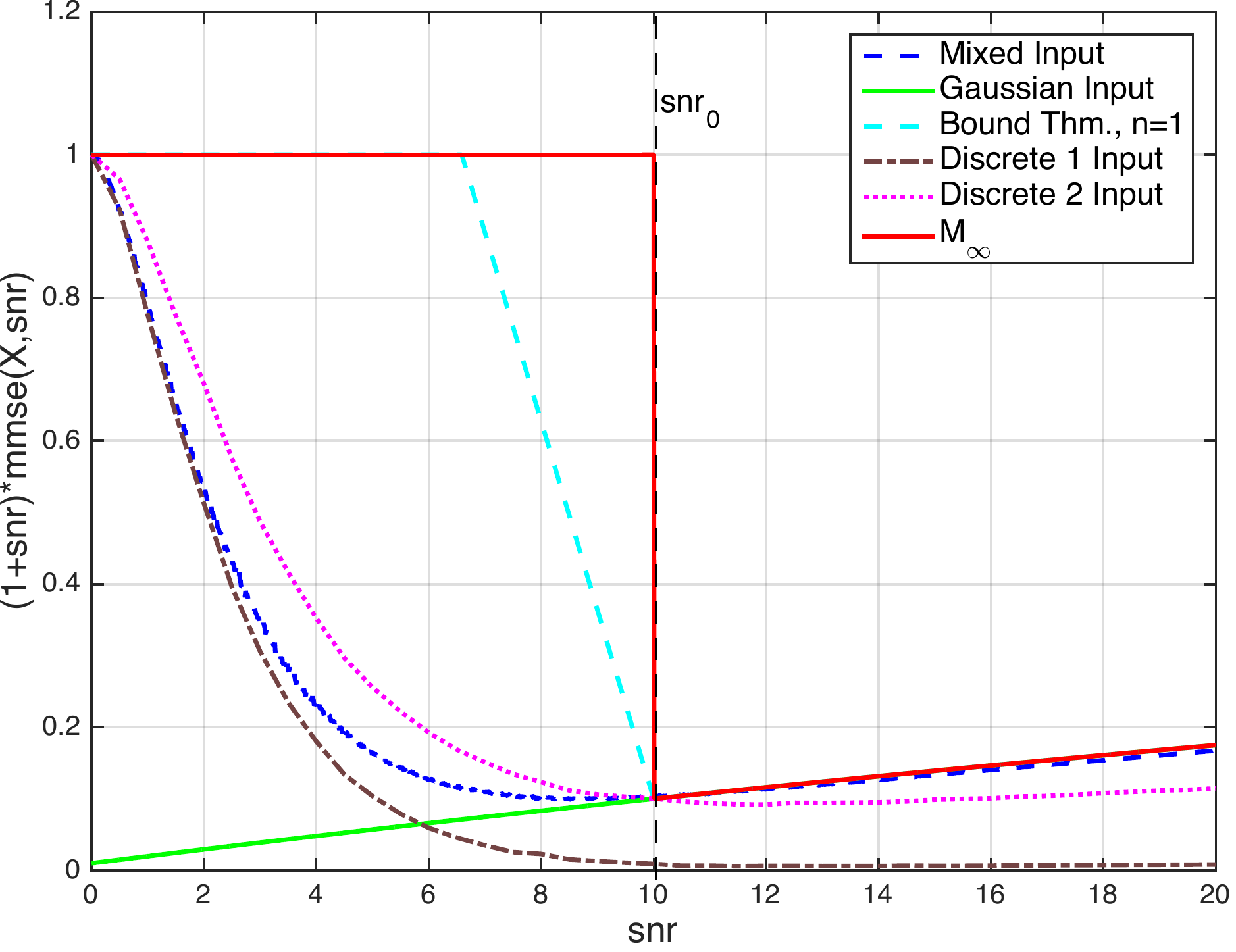}
                \caption{Upper and lower bounds on $\mathrm{M}_{1}(\snr,\snr_0,\beta)$ vs. $\snr$, for $\beta=0.01$, $\snr_0= 10$.}
                \label{fig:Mmixed}
\end{figure}

The insight gained from analyzing different lower bounds on $\mathrm{M}_{1}(\snr,\snr_0,\beta)$ will be crucial to show an approximately optimal input for $\mathcal{C}_1(\snr,\snr_0,\beta)$, which we consider next.

\subsection{max-I problem: achievability of $\mathcal{C}_1(\snr,\snr_0,\beta)$}

\begin{table*}
\centering 
\caption{Parameters of the mixed input in~\eqref{eq:input} used in the proof of Proposition~\ref{prop:gap}.}
\label{Table:input parameters}
\begin{tabular} {| c | c |}
\hline
Regime   & Input Parameters \\
\hline 
Weak Interference ($\snr \ge \snr_0$)  &  $N= \left \lfloor \sqrt{1+ c_1 \frac{(1-\delta) \snr_0}{1+\delta \snr_0}} \right \rfloor,$  $ c_1=\frac{3}{2 \log\left(\frac{12 (1-\delta) (1+\beta \snr_0)}{(1+\snr_0\delta)(\beta-\delta)}\right)}$, $\delta=\beta \frac{\snr_0}{1+\snr_0}$. \\
\hline 
Strong Interference ($ \snr \le \snr_0$) & $N=\left \lfloor \sqrt{1+ c_2  \snr }\right \rfloor$, $c_2= \frac{3}{2 \log\left(\frac{12  (1+\beta \snr_0)}{\beta }\right)}$, $\delta=0$. \\
\hline
\end{tabular}
\end{table*}

 In this section  we demonstrate that an inner bound on $\mathcal{C}_1(\snr,\snr_0,\beta)$ with the mixed input in~\eqref{eq:input} is to within  an additive gap of the outer bound in Proposition~\ref{prop: new upper bound on I with derivative}. 
\begin{prop} 
\label{prop:gap}
A lower bound on $\mathcal{C}_1(\snr,\snr_0,\beta)$ with the mixed input in~\eqref{eq:input}, with $X_D \sim \pam(N)$ and with input parameters as specified in Table~\ref{Table:input parameters}, is to within $O \left(\log\log(  \frac{1}{\mmse(X,\snr_0)} ) \right)$ of the outer bound in Proposition~\ref{prop: new upper bound on I with derivative}   with the exact gap value given by
 \begin{subequations}
\begin{align}
&
\snr \ge \snr_0 \ge 1 :C_{1}(\snr,\snr_0, \beta) -I_1(X_\mathrm{mix},\snr) := \gap_{1},\\
&
\snr_0 \ge \snr \ge 1 :C_{1}(\snr,\snr_0, \beta) -I_1(X_\mathrm{mix},\snr) := \gap_{2},\\
&
\snr \le 1 :  C_{1}(\snr,\snr_0, \beta) -I_1(X_\mathrm{mix},\snr) :=\gap_{3},
\end{align}
where 
\begin{align}
\gap_1& \le  \frac{1}{2} \log \left( \frac{2}{3} \log\left(\frac{24 (1+(1-\beta) \snr_0 }{ \beta}\right) + \frac{6\beta}{1+\beta \snr_0}  \right)+\frac{1}{2}\log \left( \frac{4 \pi }{3} \right) - \Delta_{\eqref{eq: Delta: I bound through derivative: weak interf}}, \\
\gap_2& \le  \frac{1}{2} \log \left( 1+\frac{2}{3} \log\left(\frac{12  (1+\beta \snr_0)}{\beta }\right)\right) + \frac{1}{2}\log \left(\frac{4\pi }{6} \right) -\Delta_{\eqref{eq: Delta: I bound through derivative: strong interf}}, \\
\gap_3 &\le \frac{1}{2}\log(2).
\end{align}
and $\Delta_{\eqref{eq: Delta: I bound through derivative: weak interf}}$ and $\Delta_{\eqref{eq: Delta: I bound through derivative: strong interf}}$ are given in \eqref{eq: Delta: I bound through derivative: weak interf} and \eqref{eq: Delta: I bound through derivative: strong interf}, respectively. 
\end{subequations}
\end{prop}
\begin{IEEEproof}
See Appendix~\ref{app: proof of gap}. 
\end{IEEEproof}

Please note that the gap result in  Proposition~\ref{prop:gap} is constant in $\snr$ (i.e., independent of $\snr$) but not in $\snr_0$.

Fig.~\ref{fig:gap} compares the inner bounds on $\mathcal{C}_1(\snr,\snr_0,\beta)$,
normalized by the point-to-point capacity $\frac{1}{2}\log(1+\snr)$,
with mixed inputs (dashed magenta line) in Proposition~\ref{prop:gap}  to:
\begin{itemize}
\item The   $\mathcal{C}_{\infty}(\snr,\snr_0,\beta)$ upper bound in \eqref{eq: Outer bound},  (solid red line);
\item The upper bound from Proposition~\ref{prop: new upper bound on I with derivative} (dashed blue line); and
\item The inner bound with $X \sim \mathcal{N}(0,\beta)$, where the reduction in power is necessary to satisfy the MMSE constraint $\mmse(X,\snr_0) \le \frac{\beta}{1+\beta \snr_0}$ (dotted green line). \end{itemize}
Fig.~\ref{fig:gap} shows that Gaussian inputs are sub-optimal and that mixed inputs achieve large  degrees of freedom compared to Gaussian inputs. 
Interestingly, in the regime $\snr \le \snr_0$, it is  approximately optimal to set $\delta=0$, that is, only the discrete part of the mixed input is used.  This in particular supports the conjecture in~\cite{BustinMMSEbadCodes} that discrete inputs may be optimal for $n=1$ and $\snr \le \snr_0$. For the case $\snr \ge \snr_0$ our result partially refute the conjecture by excluding the possibility of discrete inputs with finitely many points from being optimal.    

\begin{figure}
        \centering
                \includegraphics[width=8cm]{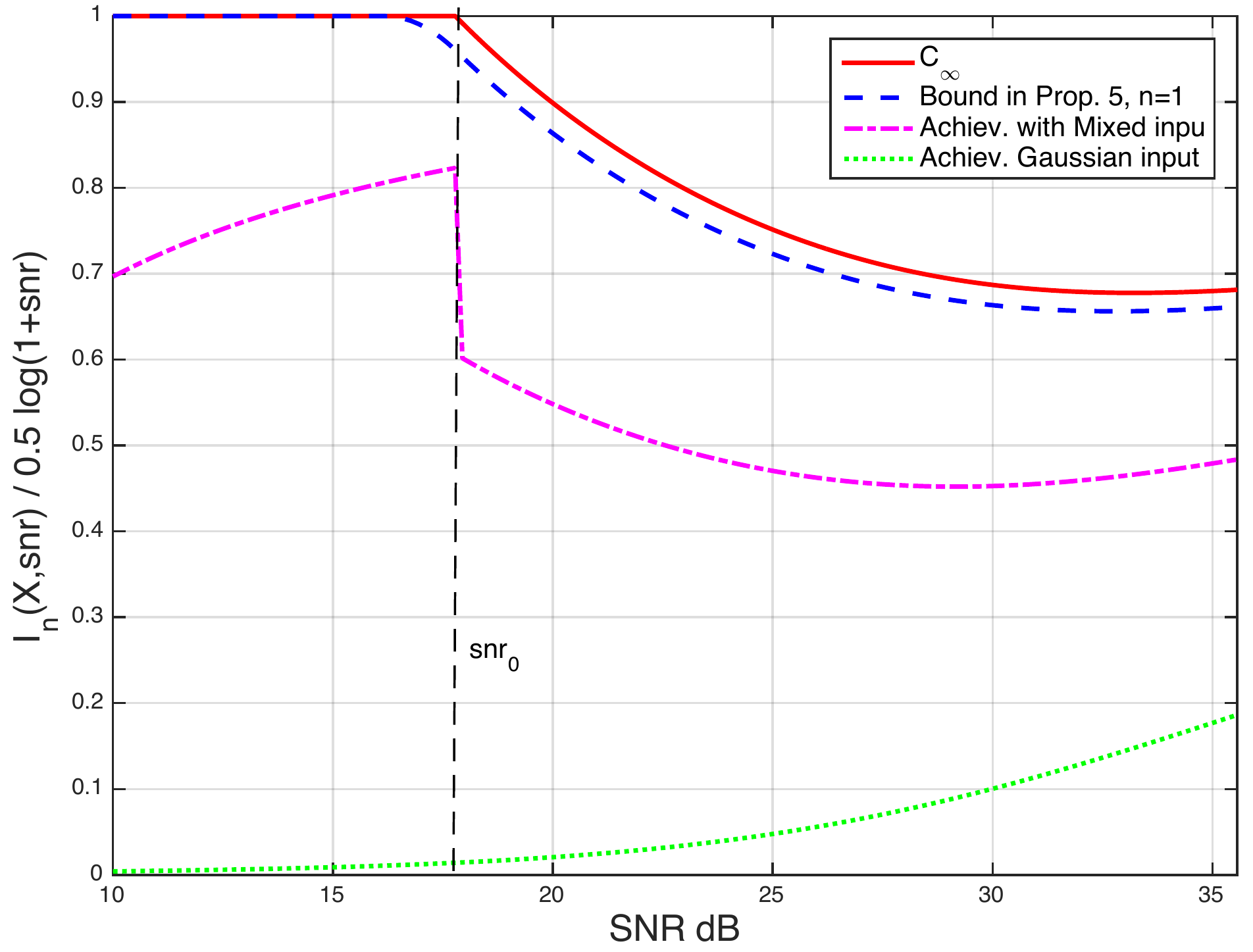}
                \caption{ Upper and lower bounds on  $\mathcal{C}_{n=1}(\snr,\snr_0,\beta)$ vs. $\snr$, for $\beta=0.001$ and $\snr_0=60= 17.6815 $ dB. }
                \label{fig:gap}
\end{figure}

 The above discussion completes the presentation of our bounds on  max-I and max-MMSE problems. The remainder of the paper contains the proof of Theorem~\ref{thm: diff MMSE} and  a  discussion of when the MMSE constraint necessarily implies a power constraint. 

\section{Properties of the first derivative of  MMSE} 
\label{sec: derivative}
A key element in the proof of  the SCPP in Proposition~\ref{prop:SCP} was the characterization of the first derivative of the MMSE as
\begin{align}
 -\frac{d \mmse(\X,\snr)}{d \snr}&
= \frac{1}{n}\Trc \left(\E \left[\cov^2(\X|\Y)\right] \right) :=\frac{1}{n}\Trc \left(\E \left[\cov^2(\X,\snr)\right] \right),
\label{eq: derivative}
\end{align}
which was given in~\cite[Proposition 9]{GuoMMSEprop}  for $n=1$  and in~\cite[Lemma 3]{BustinMMSEparallelVectorChannel}  for $n \ge 1$.
The first derivative in~\eqref{eq: derivative} turns out to be instrumental in proving Theorem~\ref{thm: diff MMSE} as well.

For  ease of  presentation, in the rest of the section, instead of focusing on the derivative  we will focus on $\Trc \left(\E[ \cov^2(\X| \Y)] \right)$. 
The quantity $\Trc \left(\E[ \cov^2(\X| \Y)] \right)$ is well defined  for any $\X$. Moreover, for the case of $n=1$ it has been shown~\cite[Proposition 5]{GuoMMSEprop}  that
\begin{align}
\E \left[\Cov^2(X|Y)\right] &\le \frac{k_1}{\snr^2}, \text{ where } k_1 \le 3 \cdot 2^4. 
\label{eq: bound on der}
\end{align}

Before using~\eqref{eq: derivative} in the proof of  Theorem~\ref{thm: diff MMSE}, we will need to sharpen  the existing constant for $n=1$ in~\eqref{eq: bound on der} (given by $k_1 \le 3 \cdot 2^4$) and generalize the bound to any $n \ge 1$, which to the best of our knowledge has not been considered before.  
\begin{prop} 
\label{prop: bound on derivative}
For any $\X$  and $\snr>0$ we have
\begin{subequations}
 \begin{align}
\frac{1}{n}\Trc \left(\E[ \cov^2(\X| \Y)]  \right)\le  \frac{ k_n }{\snr^2},  \label{eq: general bound on der}
\end{align} 
where 
\begin{align}
 k_n \le  \frac{ n(n+2)-n \ \mmse(\Z\Z^{\text{T}}|\Y)-\Trc \left( \J^2(\Y)\right)}{n}
\le n+2. \label{eq: bound on kn}
\end{align}
\label{eq: derivative bound}
\end{subequations} 
\end{prop}
 \begin{IEEEproof}
 See Appendix~\ref{app: bound on derivative}.
 \end{IEEEproof}
In Proposition~\ref{prop: bound on derivative} the bound on $k_1$ in~\eqref{eq: bound on der} has been tightened from $k_1 \le 3  \cdot 2^4$ in~\eqref{eq: bound on der} to $k_1 \le 3$. This improvement will result in tighter bounds in what follows. 

The following tightens $k_n$ for power constrained inputs. 

\begin{prop}
\label{prop: bound on FII}
If $\X$ is such that $ \frac{1}{n} \Trc \left ( \E \left[\X \X^{\text{T}} \right] \right)  \le 1$,  then\begin{align}
 \Trc( \J^2(\Y)) \ge \frac{n}{ (1+\snr )^2}. \label{eq: Bound on Fisher}
\end{align}
  Equality in \eqref{eq: Bound on Fisher} is achieved when $\X \sim \mathcal{N} ({ \bf 0}, \I)$. 
\end{prop}
\begin{IEEEproof}
 See Appendix~\ref{app:lem: bound on FII}.
\end{IEEEproof}
Observe that, by using the bound in \eqref{eq: derivative bound} from Proposition~\ref{prop: bound on derivative}  together with the lower bound on the Fisher information in Proposition~\ref{prop: bound on FII}, the bound on the constant $k_n$ in \eqref{eq: bound on kn} can be tightened to 
\begin{align}
 k_n \le  \frac{n(n+2) - \frac{n}{ (1+\snr )^2}}{n}= n+2-\frac{1}{ (1+\snr )^2}. \label{k_n: bound for finitePower}    
\end{align}

  By further assuming that $\X$ has a finite fourth moment we can arrive at the following bound that does not blow up around $\snr=0^{+}$, as opposed to the bound in \eqref{eq: general bound on der}.

\begin{prop}
\label{prop: 1/(1+snr) type bound}
If $\X$ such that $ \frac{1}{n} \Trc \left ( \E \left[ \left(\X \X^{\text{T}} \right)^2 \right] \right)  < \infty$ then
\begin{subequations} 
\begin{align}
&\Trc \left(\E[ \cov^2(\X| \Y)]  \right)  \notag\\
& \le  \min \left( \frac{\Trc \left(  \E \left[ \left( \left(\X - \sqrt{\snr}\Z \right)  \left( \X -\sqrt{\snr}\Z \right)^{\text{T}} \right)^2 \right] \right)}{(1+\snr)^4},   \Trc \left( \E \left[ \E^2 \left[ \X   \X^{\text{T}} | \Y\right]  \right] \right) \right), \label{eq: fourth moment bound}
\end{align}
where we can further bound 
\begin{align}
\Trc \left( \E \left[ \E^2 \left[ \X   \X^{\text{T}} | \Y\right]  \right] \right) \le \Trc \left ( \E \left[ \left(\X \X^{\text{T}} \right)^2 \right] \right).
\end{align}
\end{subequations}
\end{prop}
\begin{proof}
See Appendix~\ref{app:prop: 1/(1+snr) type bound}.
\end{proof}

 Note that evaluation of the first term of the minimum in \eqref{eq: fourth moment bound} requires only the knowledge of second and fourth moments of $\X$.

We are now ready to prove our main result.

\subsection{Proof of Theorem~\ref{thm: diff MMSE}  }
\label{sec: proof of diff of MMSE}
The proof of Theorem~\ref{thm: diff MMSE} relies on the fact that the MMSE is an infinitely differentiable function of  $\snr$~\cite[Proposition 7]{GuoMMSEprop} and therefore  can be written  as the difference of two MMSE functions using  the fundamental theorem of calculus 
\begin{align*}
&\mmse(\X,\snr)-\mmse(\X,\snr_0) \\
&=- \int_{\snr}^{\snr_0} \mmse'(\X,\gamma) d\gamma\\
&\stackrel{a)}{=}\int_{\snr}^{\snr_0} \frac{1}{n} \Trc \left(\E[\cov^2(\X,\gamma)] \right) d\gamma\\
& \stackrel{b)}{\le}  \int_{\snr}^{\snr_0}\frac{(n+2)}{\gamma^2} d\gamma =(n+2) \left( \frac{1}{\snr}-\frac{1}{\snr_0} \right)-\Delta, \ \Delta=0,
\end{align*}
where the (in)-equalities follow by using: 
a) \eqref{eq: derivative}, and
b) the bound in Proposition~\ref{prop: bound on derivative} with  $k_n \le n+2$.
If we further assume that  $\X$ has finite power, instead of bounding $k_n \le n+2$, we can use 
~\eqref{k_n: bound for finitePower},  
to obtain
\begin{align*}
0 \le \Delta= \Delta_{\eqref{eq:mainBoundWithPower}}&=\int_{\snr}^{\snr_0}\frac{1}{\gamma^2(1+\gamma)^2} d\gamma.
\end{align*} 
This concludes the proof of Theorem~\ref{thm: diff MMSE}.  

\section{ When does an MMSE constraint imply a power constraint} 
\label{sec:power constratint}
In this section we try to determine whether the MMSE constraint may imply a power constraint.  For simplicity we focus on the case of $n=1$.   
 This question is motivated by the following limit, which exists iff $\E[X^2]  <\infty$:
\begin{align}
\lim_{\snr \to 0^{+}} \mmse(X,\snr)=\E[X^2].\label{eq: limit  MMSE to zero}
\end{align}
The limit in \eqref{eq: limit  MMSE to zero} raises the question of whether the MMSE constraint at $\snr_0$ around zero would imply a power constraint.   In other words, are we required to reduce power to meet the MMSE constraint for very small $\snr_0$?
Surprisingly, the answer to this question is no. 
\begin{prop} 
\label{prop: counter example} 
There exists an input distribution $X$ with maximum power as in \eqref{eq:power constraint} that satisfies the MMSE constraint  in \eqref{eq:MMSE constr}
for any $\snr_0 >0$ and any $\beta>0$. 
\end{prop} 
\begin{proof}
 Consider  an input distribution given by
\begin{align}
X_a=[ -a, 0, a], \quad P_{X_a}= \left[ \frac{1}{2a^2}, 1-\frac{1}{a^2}, \frac{1}{2a^2} \right],  \label{eq: Xa counter example} 
\end{align} 
for any $a \ge 1$.
Note  that for the input  distribution in \eqref{eq: Xa counter example}  $\E[X_a^2]=1$ for any $a$.   The MMSE of $X_a$ can be upper bounded by
\begin{align}
\mmse(X_a,\snr) \le  \min \left(1, 4 (a^2+1)e^{ -\frac{a^2 \snr}{8}} \right), \label{eq:mmse of X_a upper bound}
\end{align}
where  the upper bound in \eqref{eq:mmse of X_a upper bound} follows by applying the upper bound in Proposition~\ref{prop:bounds on I and mmse of discrete}  together with the bound $\mmse(X_a,\snr) \le \E[X_a^2]=1$.
Therefore, by choosing $a$ large enough, any MMSE constraint can be met while transmitting at full power. This concludes the proof. 
\end{proof} 

The MMSE of $X_a$   is shown and compared to the LMMSE in Fig.~\ref{fig:MMSEXa}. Here are some other properties of  $X_a$ that are easy to verify.

\begin{figure}
        \centering
                \includegraphics[width=8cm]{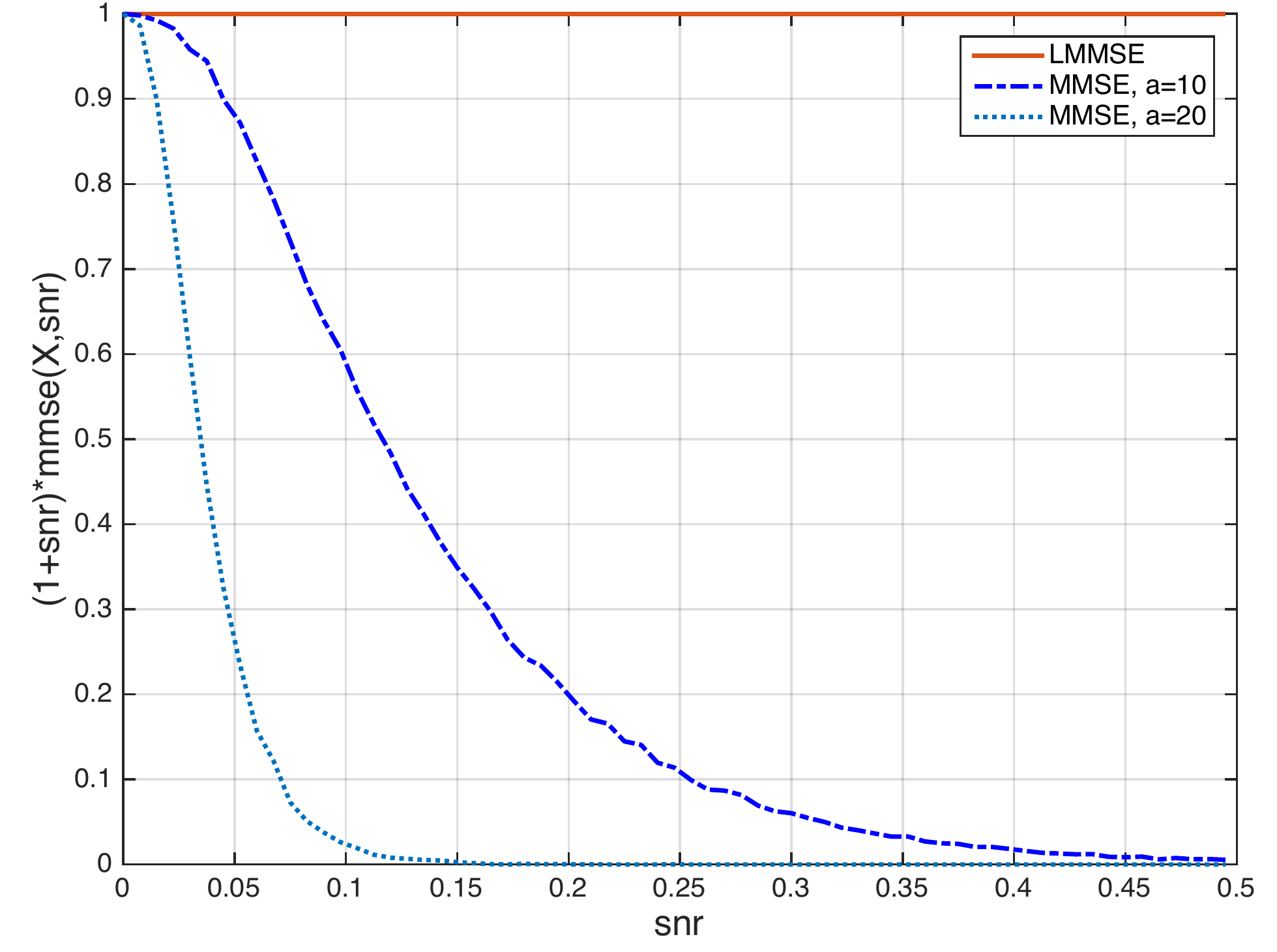}
                \caption{ $\mmse(X_a,\snr)$ vs. $\snr$, for $a=10$ and $a=20$.  }
                \label{fig:MMSEXa}
\end{figure}
\begin{prop} The random variable $X_a$ has the following properties 
\begin{itemize} 
\item $\lim_{a \to \infty} X_a =0 $ \text{ almost surely} (a.s.),
\item   $ \E[ |X_a-0|^n ]=   a^n p= E[X_a^2] a^{n-2}=a^{n-2}$.
\end{itemize}
\end{prop}

The random variable $X_a$  serves as a counterexample that shows that a.s. convergence does not imply $L^p$ convergence.  

An interesting question is whether we can characterize a family of input distributions  for which  the MMSE constraint implies a power constraint under some non-trivial condition.  In other words, we want to find a family of input  distributions  such that the  power constraint can be related to the MMSE constraint  at some $\snr_0$, that is
\begin{align}
\E[X^2]=f( \mmse(X,\snr_0)) \le 1.
\end{align}

  Towards this end we have the following:
\begin{prop} 
\label{prop:new bound for small snr0}
For any $X$  and any $\snr_0 \ge \snr>0$,  we have that 
\begin{align}
\mmse(X,\snr) = \mmse(X,\snr_ 0)+ k \cdot (\snr_0-\snr),  \label{eq: new bound for small snr0}
\end{align}
where $k$ is defined by some $\snr_c \in (\snr, \snr_0]$ as follows:
\begin{align}
k= \E \left[ \Cov^2(X,\snr_c)\right] \le \sup_{ \gamma \in (\snr,\snr_0) }  \E[ \Cov^2(X,\gamma)]  \le \E[X^4]. \label{eq: chain bounds on derivative} 
\end{align}
Moreover, for $\snr=0^{+}$ the equality in \eqref{eq: new bound for small snr0} is valid iff 
\begin{align}
\lim_{ \snr \to 0^{+}} \E[ \Cov^2(X,\snr)]  <\infty. 
\end{align} 
\end{prop}
\begin{proof}
The result easily follows by applying the mean value theorem 
\begin{align}
\mmse(X,\snr) -\mmse(X,\snr_0) &= \int_{\snr}^{\snr_0}    \E[ \Cov^2(X,\gamma)] d \gamma  \notag\\
&= \E[ \Cov^2(X,\snr_c)] (\snr_0-\snr).  \label{eq: MVT}
\end{align}
for some $\snr_c \in (\snr,\snr_0)$.
Note that  for $\snr>0$ the quantity $ \E[ \Cov^2(X,\gamma)] $ is finite due to Proposition~\ref{prop: bound on derivative}.   Therefore, we focus on the case when $\snr=0^{+}$. 

  Therefore, if $ \lim_{ \snr \to 0^{+}} \E[ \Cov^2(X,\snr)]  =K <\infty$ for  some $K>0$,  by Jensen's inequality we have that
\begin{align}
 K=\lim_{ \snr \to 0^{+}} \E[ \Cov^2(X,\snr)] \ge  (\E[X^2])^2=  (\mmse(X,0))^2.
\end{align}
So, in other words the existence of the derivative at $\snr=0^{+}$ implies the existence of the power constraint and the integration  in  \eqref{eq: MVT} holds for $\snr=0^{+}$.

Conversely, if the integration  in  \eqref{eq: MVT} is finite for $\snr=0^{+}$ we have that     \\ $\lim_{ \snr \to 0^{+}} \E[ \Cov^2(X,\snr)]  <\infty$. 

Therefore,  the bound in \eqref{eq: new bound for small snr0} holds iff  $\lim_{ \snr \to 0^{+}} \E[ \Cov^2(X,\snr)]  <\infty$. This concludes the proof. 
\end{proof}

From Proposition~\ref{prop:new bound for small snr0} we see that necessary and sufficient conditions for the MMSE at $\snr_0$ to imply a reduction in  power (i.e., $\E[X^2]<1$)  are
\begin{subequations}
 \begin{align}
1) \ &\mmse(X,\snr_0)+\snr_0 \cdot \E[ \Cov^2(X,\snr_c)]  < 1, \notag \\
& \Leftrightarrow  \E[ \Cov^2(X,\snr_c)]< \frac{1-\mmse(X,\snr_0)}{\snr_0},\\
2) \ &\lim_{ \snr \to 0^{+}} \E[ \Cov^2(X,\snr)]  <\infty,
\end{align}
\end{subequations} 
where $\snr_c$ is defined in Proposition~\ref{prop:new bound for small snr0}. 

Since $\snr_c$ might be difficult to compute, the following slightly stronger (i.e., sufficient condition) can be  useful:
\begin{align}
   \sup_{\gamma \in (0,\snr_0)}\E[ \Cov^2(X,\gamma)]< \frac{1-\mmse(X,\snr_0)}{\snr_0}. 
\end{align}

Finally, observe that $\lim_{ a \to \infty} X_a$ does not satisfy this moment condition since  
\begin{align}
\lim_{ a \to \infty}  \E[\Cov^2(X_a|Y)]&= \left \{ 
\begin{array}{ll}  \infty  & \snr=0, \\
0  & \snr>0. \end{array} \right.
\end{align}

\section{Conclusion}
In this paper we have considered a Gaussian channel with one transmitter and two receivers in which the maximization of the  input-output mutual information  at the primary/intended receiver is subject to a disturbance constraint measured by the MMSE at the secondary/unintended receiver.
We have  derived new upper bounds on the  input-output mutual information of this channel that hold for vector inputs of any length.  For the case of scalar inputs we have demonstrated  a matching lower bound that is to within an additive gap of the order $O \left( \log \log \frac{1}{\mmse(X,\snr_0)} \right)$ of the  upper bound.
At the heart of our proof  is  a new upper bound on the MMSE that complements the SCPP of the MMSE and  might be of independent interest.

\appendices

\section{Proof of Proposition~\ref{prop:snrL:diff MMSE bound}}
\label{app: snrL for D-bound}
In order to find the point of intersection $\snr_L$ between \eqref{eq: LMMSE decorrolator} and \eqref{eq:mainBound} we must solve the following equation:
\begin{align*}
\frac{1}{\snr}-\frac{k_n}{\snr}+\frac{k_n}{\snr_0}-\frac{\beta}{1+\beta \snr_0}=0 \Rightarrow \frac{1}{\snr}-\frac{k_n}{\snr}+A=0
\end{align*}
where $A=\frac{k_n}{\snr_0}-\frac{\beta}{1+\beta \snr_0}$  contains all quantities that do not depend on $\snr$.  By solving for $\snr$ we find that
\begin{align*}
\snr_L&=\frac{ k_n-1}{A}=\frac{\snr_0(1+\beta \snr_0) (k_n-1)}{k_n+(k_n-1) \beta \snr_0}= \snr_0\frac{1+\beta \snr_0}{\frac{k_n}{k_n-1} +\beta \snr_0},
\end{align*}
and the width of the phase transition is given by
\begin{align*}
\snr_0-\snr_L&=\snr_0 \left( 1-\frac{1+\beta \snr_0 }{ \frac{k_n}{k_n-1}+ \beta \snr_0}\right)=   \frac{1}{k_n-1}\frac{\snr_0}{ \frac{k_n}{k_n-1}+\beta \snr_0},
\end{align*}
as claimed in \eqref{eq:width phase trans}. 
This concludes the proof. 

\section{Proof of Proposition~\ref{prop: decomposition}}
\label{app: proof prop decomposition}
We first show the decomposition for mutual information  with mixed inputs in \eqref{eq:input}
\begin{align}
I(X_\mathrm{mix},\snr)&=I(X_\mathrm{mix};Y)=  I(X_G,X_D;Y) \notag\\
&=I(X_D;Y)+I(X_G;Y|X_D)\notag\\
&= I\left(X_D, \frac{\snr(1-\delta)}{1+\delta \snr}\right)+I(X_G, \snr \delta ). \label{eq: decompostion MI}
\end{align}

Next we take the derivative of both sides of \eqref{eq: decompostion MI} with respect to $\snr$. On the left side we get $\frac{d}{d \snr}I(X_\mathrm{mix},\snr)= \frac{1}{2}\mmse(X_\mathrm{mix},\snr)$ and on the right we get
\begin{align*}
&\mmse(X_\mathrm{mix},\snr)\\
&=2\frac{d}{d\snr}I\left(X_D, \frac{\snr(1-\delta)}{1+\delta \snr}\right)+2\frac{d}{d\snr} I(X_G, \snr\delta )\\
& = \mmse\left(X_D, \frac{\snr(1-\delta)}{1+\delta \snr}\right) \cdot \frac{d}{d\snr}\left(\frac{\snr(1-\delta)}{1+\delta \snr}\right) +   \mmse(X_G, \snr\delta ) \cdot  \frac{d}{d\snr}  \left(\snr\delta  \right)\\
&   =\frac{1-\delta}{(1+\delta \snr)^2}  \mmse\left(X_D, \frac{\snr(1-\delta)}{1+\delta \snr}\right)+  \mmse(X_G, \snr\delta )  \delta \\\
& =  \frac{1-\delta}{(1+\delta \snr)^2}  \mmse\left(X_D, \frac{\snr(1-\delta)}{1+\delta \snr}\right)+  \frac{\delta}{1+\delta \snr},
\end{align*}
as claimed in \eqref{eq: decompositions of Mixed}. 
This concludes the proof.

\section{Proof of Proposition~\ref{prop:gap}}
\label{app: proof of gap}

By  letting $X_D \sim \pam(N)$,  given the bound in Proposition~\ref{prop:bounds on I and mmse of discrete} and the requirement in \eqref{eq: constraint on discrete} we further constrain the MMSE of $X_D$ to satisfy
\begin{align}
\mmse\left(X_D, \frac{\snr_0(1-\delta)}{1+\delta \snr_0}\right) \le d_{\max}^2  \eu^{-\frac{\frac{\snr_0(1-\delta)}{1+\delta \snr_0}}{8} d_{\min}^2} \le \frac{(1+\snr_0 \delta)(\beta-\delta)}{(1-\delta) (1+\beta \snr_0)},  \label{eq: bound to satisify}
\end{align}
which ensures that the MMSE  constraint in \eqref{eq:MMSE constr} is  met.  Since, the minimum distance  of PAM is given by  $d_{\min}^2=\frac{12}{N^2-1}$, solving for $N$ we have that
\begin{subequations}
\begin{align}
N &\le \left \lfloor \sqrt{1+ c_1 \frac{(1-\delta) \snr_0}{1+\delta \snr_0}} \right \rfloor,\\
c_1&= \frac{3}{2 \log^{+}\left(\frac{d_{\max}^2 (1-\delta) (1+\beta \snr_0)}{(1+\snr_0\delta)(\beta-\delta)}\right)} \le \frac{3}{2 \log^{+}\left(\frac{12 (1-\delta) (1+\beta \snr_0)}{(1+\snr_0\delta)(\beta-\delta)}\right)} ,
\end{align}
\label{eq: input parameters for weak}
\end{subequations}
where the last inequality is due to the fact that for PAM  
\begin{align}
d_{\max}^2 = (N-1)^2 d_{\min}^2 =12 \frac{(N-1)^2}{N^2-1} =12 \frac{N-1}{N+1} \le 12.
\end{align}

For the case of $\snr_0 \le \snr$  we choose the number of points to satisfy \eqref{eq: input parameters for weak} with equality and choose $\delta=\beta \frac{\snr_0}{1+\snr_0}:=\beta c_2$. 

Next we compute the gap between the outer bound in Proposition~\ref{prop: new upper bound on I with derivative}   with the achievable mutual information of a mixed input in Proposition~\ref{prop: decomposition}, where $I \left(X_D, \frac{\snr(1-\delta)}{1+\delta \snr} \right)$ is lower bounded  by  Proposition~\ref{prop:bounds on I and mmse of discrete}  we have

We obtain
\begin{align}
&\gap_1+\Delta_{\eqref{eq: Delta: I bound through derivative: weak interf}}  \notag\\
&=  \mathcal{C}_{\infty}-I \left(X_D, \frac{\snr(1-\delta)}{1+\delta \snr} \right) - I(X_G, \snr \ \delta ) \\
&=C_{\infty}-  \left(\log(N) -\frac{1}{2} \log \left( \frac{\pi }{6}\right)-\frac{1}{2}\log \left( 1+\frac{12}{ d^2_{\min}} \mmse \left(X_D,\frac{\snr(1-\delta)}{1+\delta \snr}\right)\right) \right) -\frac{1}{2} \log(1+\delta \snr)
\notag\\
& \stackrel{a)}{\le} C_{\infty}-  \left(\frac{1}{2} \log\left(1+ c_1 \frac{(1-\delta) \snr_0}{1+\delta \snr_0}\right)-\log(2)  -\frac{1}{2} \log \left( \frac{\pi }{6}\right) \right.\notag\\
& \left.-\frac{1}{2}\log \left( 1+\frac{12}{ d^2_{\min}} \mmse \left(X_D,\frac{\snr(1-\delta)}{1+\delta \snr}\right)\right)+\frac{1}{2} \log(1+\delta \snr) \right)\notag\\
&=  \frac{1}{2}\log \left( \frac{1+\frac{\snr_0(1-\beta)}{1+\beta\snr_0}}{1 +c_1 \frac{ (1-\delta)\snr_0}{1+\delta \snr_0}}\right)+\frac{1}{2} \log\left(\frac{1+\beta \snr}{1+\delta \snr}\right)+\frac{1}{2}\log \left( 1+\frac{12}{ d^2_{\min}} \mmse \left(X_D,\frac{\snr(1-\delta)}{1+\delta \snr}\right)\right) \notag\\
&+\frac{1}{2}\log \left(\frac{4\pi }{6} \right),
 \label{eq: gap expression not simplified yet}
\end{align}
where inequality in a) follows from getting an extra one bit gap from dropping the floor operation.

We next bound each term in \eqref{eq: gap expression not simplified yet} individually.
The first term  in \eqref{eq: gap expression not simplified yet} can be bounded  as follows:
\begin{align}
\frac{1}{2}\log \left( \frac{1+\frac{\snr_0(1-\beta)}{1+\beta\snr_0}}{1+c_1 \frac{ (1-\delta)\snr_0}{1+\delta \snr_0}}\right)&=\frac{1}{2}\log \left( \frac{(1+\snr_0) (1+c_2 \beta \snr_0)}{(1+
\beta \snr_0)(1+c_1 \snr_0+\beta c_2 \snr_0-\beta c_1 c_2\snr_0)} \right) \notag\\
&\stackrel{b)}{\le} \frac{1}{2}\log \left( \frac{(1+\snr_0) (1+c_2 \beta \snr_0)}{(1+
\beta \snr_0)(1+c_1 \snr_0+\beta c_2 \snr_0-\beta c_1  \snr_0)} \right) \notag\\
&=\frac{1}{2}\log \left( \frac{(1+\snr_0) (1+c_2 \beta \snr_0)}{(1+
\beta \snr_0)(1+(1-\beta)c_1 \snr_0+\beta c_2 \snr_0)} \right)\notag\\
&\stackrel{c)}{\le} \frac{1}{2}\log \left( \frac{(1+\snr_0)}{(1+(1-\beta)c_1 \snr_0+\beta c_2 \snr_0)} \right)\notag\\
& \stackrel{d)}{\le} \frac{1}{2}\log \left(  \max \left(\frac{ (1+\snr_0)}{(1+c_1 \snr_0)},\frac{ (1+\snr_0)}{(1+ c_2 \snr_0)} \right) \right)\notag\\
& \stackrel{e)}{\le} \frac{1}{2} \log \left( \max \left( \frac{1}{c_1}, 2\right) \right),
\label{eq: bound on the first term}
\end{align}
where the inequalities follow from the facts: b)   $c_2=\frac{\snr_0}{1+\snr_0} \le 1$; c) used that $\frac{1+c_2 \beta \snr_0}{1+
\beta \snr_0}\le 1$  since $c_2 \le 1$; d) the denominator term $1+(1-\beta)c_1 \snr_0+\beta c_2 \snr_0$ achieves its minimum at either $\beta=0$ or $\beta=1$; and e) $\frac{ (1+\snr_0)}{(1+ c_2 \snr_0)} \le  \frac{1}{c_2} =\frac{1+\snr_0}{\snr_0}\le 2$ for  $ \snr_0 \ge 1$.

The second term in \eqref{eq: gap expression not simplified yet} can be bounded as follows:
\begin{align} 
&\frac{1}{2} \log\left(\frac{1+\beta \snr}{1+\delta \snr}\right) \le \frac{1}{2} \log\left( \frac{1+\snr_0}{\snr_0}\right) \le \frac{1}{2} \log\left(2\right),  \label{eq: bound second term}
\end{align}
where the inequalities follow from using $\delta=\beta \frac{\snr_0}{1+\snr_0}$ and $\frac{1+\beta \snr}{1+\delta \snr} \le \frac{\beta}{\delta}= \frac{1+\snr_0}{\snr_0} \le 2$ for  $\snr \ge \snr_0 \ge 1$.

The third term in \eqref{eq: gap expression not simplified yet} can be bounded as follows
\begin{align}
&\frac{1}{2}\log \left( 1+\frac{12}{ d^2_{\min}} \mmse \left(X_D,\frac{\snr(1-\delta)}{1+\delta \snr}\right)\right) \notag\\
 &\stackrel{f)}{\le}  \frac{1}{2}\log \left( 1+\frac{12}{ d^2_{\min}} \mmse \left(X_D,\frac{\snr_0(1-\delta)}{1+\delta \snr_0}\right)\right) \notag\\
& \stackrel{g)}{\le}  \frac{1}{2}\log \left( 1+c_1\frac{(1-\delta) \snr_0}{1+\delta \snr_0} \mmse \left(X_D,\frac{\snr_0(1-\delta)}{1+\delta \snr_0}\right)\right) \notag\\
& \stackrel{h)}{\le}  \frac{1}{2}\log \left( 1+c_1 \frac{(\beta-\delta)\snr_0}{1+\beta \snr_0}\right)\notag\\
& \stackrel{i)}{\le}   \frac{1}{2}\log \left( 1+c_1 \frac{\beta}{1+\beta \snr_0}\right) \label{eq: bound on third term},
\end{align}
where the (in)-equalities follow from: f) the fact that the MMSE is a decreasing function of SNR and  $\frac{\snr(1-\delta)}{1+\delta \snr} \ge \frac{\snr_0(1-\delta)}{1+\delta \snr_0}$; g) using the bound on $d_{\min}^2=\frac{12}{N^2-1}$ from \eqref{eq: input parameters for weak}; h) using the bound in \eqref{eq: bound to satisify}; and  i) using $\delta=\frac{\beta \snr_0}{1+\snr_0} \le \beta$  and therefore $(\beta-\delta) \snr_0 =\frac{\beta \ \snr_0 }{1+\snr_0} \le \beta$.

By combining the bounds in  \eqref{eq: bound on the first term}, \eqref{eq: bound second term}, and \eqref{eq: bound on third term} we get 
\begin{align*}
\gap_2+\Delta_{\eqref{eq: Delta: I bound through derivative: weak interf}}  &\le\frac{1}{2} \log \left( \max \left( \frac{1}{c_1}, 2\right) \right)+\frac{1}{2}\log \left( \frac{4 \pi }{3} \right)+\frac{1}{2}\log \left( 1+c_1 \frac{\beta}{1+\beta \snr_0}\right)\\
&=\frac{1}{2} \log \left( \max \left( \frac{1}{c_1}, 2\right) +2 \max \left( 1, 2c_1 \right) \frac{\beta}{1+\beta \snr_0} \right)+\frac{1}{2}\log \left( \frac{4 \pi }{3} \right)\\
&\stackrel{j)}{\le} \frac{1}{2} \log \left( \max \left( \frac{1}{c_1}, 2\right) +6 \frac{\beta}{1+\beta \snr_0} \right)+\frac{1}{2}\log \left( \frac{4 \pi }{3} \right)\\
&\stackrel{k)}{=} \frac{1}{2} \log \left( \max \left( \frac{2 \log\left(\frac{12 (1-\delta) (1+\beta \snr_0)}{(1+\snr_0\delta)(\beta-\delta)}\right)}{3}, 2\right) +6 \frac{\beta}{1+\beta \snr_0} \right)+\frac{1}{2}\log \left( \frac{4 \pi }{3} \right)\\
& \stackrel{l)}{\le}  \frac{1}{2} \log \left( \max \left(\frac{2}{3} \log\left(\frac{24 (1+(1-\beta) \snr_0 }{ \beta}\right), 2\right) +6 \frac{\beta}{1+\beta \snr_0}  \right)+\frac{1}{2}\log \left( \frac{4 \pi }{3} \right)\\
& \stackrel{m)}{=}  \frac{1}{2} \log \left( \frac{2}{3} \log\left(\frac{24 (1+(1-\beta) \snr_0 }{ \beta}\right) +6 \frac{\beta}{1+\beta \snr_0}  \right)+\frac{1}{2}\log \left( \frac{4 \pi }{3} \right),
\end{align*}
where the inequalities follow from: j) the fact that $c_1 \le \frac{3}{2}$;  k) using the value of $c_1$ in \eqref{eq: input parameters for weak}; l) using $\delta=\beta \frac{\snr_0}{1+\snr_0}$ and $\frac{1+\beta \snr_0}{1+\delta \snr_0} \le \frac{1+\snr_0}{\snr_0} \le 2$ for $\snr_0 \ge 1$; and m) the fact that $\max \left(\frac{2 \log\left(\frac{24 (1+\beta \snr_0)}{\beta }\right)}{3}, 2\right)=\frac{2 \log\left(\frac{24 (1+\beta \snr_0)}{\beta }\right)}{3}$.

This concludes the proof of the gap result for the $\snr \ge \snr_0$ regime. 

We next focus on the  $1 \le \snr \le \snr_0$ regime.  We  use only the discrete part  of the  mixed input and set $\delta=0$. From \eqref{eq: input parameters for weak}  we have that the input parameters  must satisfy
\begin{subequations}
\begin{align}
N &\le \left \lfloor \sqrt{1+ c_3  \snr_0 }\right \rfloor, \\
c_3&\le \frac{3}{2 \log\left(\frac{12  (1+\beta \snr_0)}{\beta }\right)} , 
\end{align}
\label{eq: number of points: strong}
\end{subequations}
in order to comply with the MMSE constraint in \eqref{eq:MMSE constr}.  However, instead of  choosing the number of points as in \eqref{eq: number of points: strong} we choose it to be 
\begin{align}
N &= \left \lfloor \sqrt{1+ c_3  \snr }\right \rfloor \le \left \lfloor \sqrt{1+ c_3  \snr_0 }\right \rfloor. \label{eq: choice: for strong}
\end{align}
The reason for this choice will be apparent from the gap derivation next. 

Similarly to the previous case, we compute the gap between the outer bound in Proposition~\ref{prop: new upper bound on I with derivative}   and the achievable  mutual information of the mixed input in Proposition~\ref{prop: decomposition}, where $I \left(X_D, \snr \right)$ is lower bounded  using  Proposition~\ref{prop:bounds on I and mmse of discrete}. We have,
\begin{align*}
\gap_2+\Delta_{\eqref{eq: Delta: I bound through derivative: strong interf}}  &\le C_{\infty}  -   \log(N)+\frac{1}{2}\log \left(\frac{\pi \eu}{6} \right)+\frac{1}{2}\log \left(1+\frac{12}{d_{\min}^2} \mmse(X_D,\snr)\right) \\
& \stackrel{n)}{ \le}  \frac{1}{2} \log \left( \frac{1+\snr}{1+c_3 \snr}\right)+  \frac{1}{2}\log \left(\frac{4\pi \eu}{6} \right)+\frac{1}{2}\log \left(1+\frac{12}{d_{\min}^2} \mmse(X_D,\snr) \right)\\
& \stackrel{o)}{ \le} \frac{1}{2} \log \left( \frac{1+\snr}{1+c_3 \snr}\right) + \frac{1}{2}\log \left(\frac{4\pi \eu}{6} \right)+\frac{1}{2}\log \left(1+\frac{ c_3 \snr}{1+\snr} \right)\\
& =\frac{1}{2} \log \left( \frac{1+(1+c_3)\snr}{1+c_3 \snr}\right) + \frac{1}{2}\log \left(\frac{4\pi \eu}{6} \right)\\
& \stackrel{p)}{\le} \frac{1}{2} \log \left( 1+\frac{1}{c_3}\right) + \frac{1}{2}\log \left(\frac{4\pi \eu}{6} \right)\\
&  \stackrel{r)}{=} \frac{1}{2} \log \left( 1+\frac{2}{3} \log\left(\frac{12  (1+\beta \snr_0)}{\beta }\right)\right) + \frac{1}{2}\log \left(\frac{4\pi \eu}{6} \right),
\end{align*}
where the  (in)-equalities follow from: n)  getting an extra one bit gap by dropping the floor operation; o) using the bound on $d_{\min}^2=\frac{12}{N^2-1}$ from \eqref{eq: choice: for strong} and bound $\mmse(X,\snr) \le \frac{1}{1+\snr}$; p) using  that $\frac{1+(1+c_3)\snr}{1+c_3 \snr} \le \frac{1+c_3}{c_3}=1+\frac{1}{c_3}$; and r) using the value of $c_3$ from \eqref{eq: number of points: strong}.

 This concludes the proof for the case $1 \le \snr \le \snr_0$.

 Finally, note that for the case $\snr \le 1$ the gap is trivially given by 
\begin{align}
\gap_3 \le \mathcal{C}(\beta,\snr,\snr_0)-I (X_{\text{mix}},\snr) \le   \mathcal{C}(\beta,\snr,\snr_0) \le \frac{1}{2}\log(1+\snr) \le \frac{1}{2}\log(2).
\end{align} 

This concludes the proof. 

\section{ Proof of Proposition~\ref{prop: bound on derivative}}
\label{app: bound on derivative}
 We will need the following  identities for the proof:
\begin{subequations}
\begin{align}
\snr \cdot  \E[  \cov(\X|\Y) ]=  \E[\cov(\Z|\Y)],\\
\snr^2 \cdot  \E[  \cov^2(\X|\Y) ]=  \E[\cov^2(\Z|\Y)],
\end{align}
\label{eq:cov identiies}
\end{subequations}
which follow since 
\begin{align*}
\sqrt{\snr}\X+\Z= \Y =\E[\Y| \Y]= \sqrt{\snr} \E[\X| \Y]+\E[\Z| \Y],
\end{align*}
and therefore
\begin{align*}
\sqrt{\snr} (\X-\E[\X| \Y] )= (\Z-\E[\Z| \Y]). 
\end{align*}

Next, Observe that
 \begin{align*}
 \cov(\Z|\Y)=  \E[\Z\Z^{\text{T}}|\Y]-(\E[\Z|\Y])(\E[\Z|\Y])^{\text{T}},
 \end{align*}
 and so we have that 
 \begin{align}
 \cov^2(\Z| \Y) &=\left(\E[\Z \Z^{\text{T}}|\Y]-\E[\Z|\Y] \E[\Z|\Y]^{\text{T}} \right)^2 \notag\\
 &=(\E[\Z \Z^{\text{T}}|\Y])^2-\E[\Z|\Y] \E[\Z|\Y]^{\text{T}} \E[\Z \Z^{\text{T}}|\Y] \notag\\
 &-\E[\Z \Z^{\text{T}}|\Y]\E[\Z|\Y] \E[\Z|\Y]^{\text{T}}+(\E[\Z|\Y] \E[\Z|\Y]^{\text{T}})^2 \notag \\
 &\stackrel{a)}{=}(\E[\Z \Z^{\text{T}}|\Y])^2-2\E[\Z|\Y] \E[\Z|\Y]^{\text{T}} \E[\Z \Z^{\text{T}}|\Y]\ \notag\\
 &+(\E[\Z|\Y] \E[\Z|\Y]^{\text{T}})^2 \notag \\
 &\stackrel{b)}{\preceq} (\E[\Z \Z^{\text{T}}|\Y])^2-2\E[\Z|\Y] \E[\Z|\Y]^{\text{T}} \E[\Z|\Y] \E[\Z|\Y]^{\text{T}} \notag\\
 &+(\E[\Z|\Y] \E[\Z|\Y]^{\text{T}})^2 \notag\\
 & = (\E[\Z \Z^{\text{T}}|\Y])^2-(\E[\Z|\Y] \E[\Z|\Y]^{\text{T}})^2 \notag \\
  &\stackrel{c)}{=}\E[\Z \Z^{\text{T}} (\Z \Z^{\text{T}})^{\text{T}}|\Y]-\cov(\Z\Z^{\text{T}}|\Y)-(\E[\Z|\Y] \E[\Z|\Y]^{\text{T}})^2,\label{eq: cond var inequaliyt}
 \end{align}
 where  the order operations follow from:  a) the fact that $\E[\Z|\Y] \E[\Z|\Y]^{\text{T}}$ and $\E[\Z \Z^{\text{T}}|\Y]$ are symmetric matrices; b) using   $\E[\Z|\Y] \E[\Z|\Y]^{\text{T}} \preceq \E[\Z \Z^{\text{T}}|\Y]$  (from the positive semi-definite property of the conditional covariance matrix); and c) the fact that, since  $\cov(\Z\Z^{\text{T}}|\Y)=\E[\Z \Z^{\text{T}} (\Z \Z^{\text{T}})^{\text{T}}|\Y]-\E[\Z \Z^{\text{T}}|\Y](\E[\Z \Z^{\text{T}}|\Y])^{\text{T}}$ and by symmetry  of $\E[\Z \Z^{\text{T}}|\Y]$, we have that  \\ $\E[\Z \Z^{\text{T}}|\Y](\E[\Z \Z^{\text{T}}|\Y])^{\text{T}}= (\E[\Z \Z^{\text{T}}|\Y])^2$.
By using the monotonicity of the trace, properties of the expected value, and the  inequality in~\eqref{eq: cond var inequaliyt}, we have that
\begin{align}
\Trc \left(\E[ \cov^2(\Z| \Y)] \right) &\le \Trc \left(  \E \left[ \E[\Z \Z^{\text{T}} (\Z \Z^{\text{T}})^{\text{T}}|\Y]-\cov(\Z\Z^{\text{T}}|\Y) -(\E[\Z|\Y] \E[\Z|\Y]^{\text{T}})^2 \right] \right) \notag\\
&=\Trc \left(  \E \left[ \E[\Z \Z^{\text{T}} (\Z \Z^{\text{T}})^{\text{T}}|\Y]\right] \right)   - \Trc \left(  \E \left[\cov(\Z\Z^{\text{T}}|\Y) \right] \right)\notag \\
& -\Trc \left(  \E \left[(\E[\Z|\Y] \E[\Z|\Y]^{\text{T}})^2 \right] \right).  \label{eq: trace inequality}
\end{align}
We next focus on each term of the right hand side of~\eqref{eq: trace inequality} individually. The first term can be computed as follows:
\begin{align}
\Trc \left( \E \left[  \E[\Z \Z^{\text{T}} (\Z \Z^{\text{T}})^{\text{T}}|\Y]  \right]\right)&\stackrel{d)}{=}\Trc \left(  \E[\Z \Z^{\text{T}} \Z \Z^{\text{T}}] \right)\notag\\
&\stackrel{e)}{=}\E\left[\Trc \left( \Z \Z^{\text{T}} \Z \Z^{\text{T}}\right)\right] \notag\\
&=\E\left[\Trc \left(\Z^{\text{T}} \Z \Z^{\text{T}} \Z \right)\right]\notag\\
&=\E\left[ \left(\sum_{i=1}^n Z_i^2 \right)^2 \right] \notag \\
&\stackrel{f)}{=} n(n+2), \label{eq: eq1: trace inequality}
\end{align}
where the (in)-equalities follow from: d) using the law of total expectation; 
e) since expectation is a linear operator and  using fact that the trace can be exchanged with linear operators; and 
f) observing that $S=\sum_{i=1}^n Z_i^2 $ is a chi-square distribution of degree $n$ and hence $\E[S]=n(n+2)$. 

For the second term  in ~\eqref{eq: trace inequality},  by definition of the MMSE, we have 
\begin{align} 
\Trc \left(  \E \left[\cov(\Z\Z^{\text{T}}|\Y) \right] \right) = n \mmse(\Z\Z^{\text{T}}|\Y). \label{eq: eq2: trace inequality}
\end{align}

The third term in ~\eqref{eq: trace inequality} satisfies
\begin{align}
\Trc \left(  \E \left[(\E[\Z|\Y] \E[\Z|\Y]^{\text{T}})^2 \right] \right)  &\stackrel{g)}{\ge} \Trc \left(  \left(\E \left[ \E[\Z|\Y] \E[\Z|\Y]^{\text{T}} \right] \right)^2 \right) \notag\\
&= \Trc \left(  \left(\E[\Z\Z^{\text{T}}]-\E[\cov(\Z|\Y)]\right)^2 \right)\notag\\
&\stackrel{h)}{=} \Trc \left(  \left(\I-\snr \ \E[\cov(\X|\Y)]\right)^2 \right) \notag\\
&\stackrel{i)}{=}  \Trc \left( \J^2(\Y)\right) \label{eq: eq3: trace inequality}
\end{align}
 where the (in)-equalities follow from: g) using Jensen's inequality; h)  using the property: $\snr  \cdot  \,\E[\cov(\X|\Y)]= \E[\cov(\Z|\Y)]$ in \eqref{eq:cov identiies}; and i) using identity \cite{GuoMMSEprop}
 \begin{align*}
 \I-\snr \ \E[\cov(\X|\Y)]=\J(\Y).
 \end{align*}

By putting~\eqref{eq: eq1: trace inequality}, \eqref{eq: eq2: trace inequality}, and~\eqref{eq: eq3: trace inequality} together, we have that
 \begin{align*}
\E\left[\cov^2(\Z|\Y)\right] \le k_n :=  \frac{ n(n+2)-n \ \mmse(\Z\Z^{\text{T}}|\Y)-\Trc \left( \J^2(\Y)\right)}{n}.
 \end{align*}
 Finally, using the identity $\E\left[\cov^2(\Z|\Y)\right]=\snr^2 \cdot\E\left[\cov^2(\X|\Y)\right]$ in \eqref{eq:cov identiies} concludes the proof.

\section{Proof of  Proposition~\ref{prop: bound on FII}}
\label{app:lem: bound on FII}
Using the Cramer-Rao lower bound \cite[Theorem 20]{DemboCoverInfoInequalities} we have that
\begin{align*}
\J(\Y) &\succeq   \cov^{-1}(\Y) \\
&= \left( \snr \E[\X\X^T] +\I \right)^{-1}\\
&=  {\bf V}^{-1} { \bf \Lambda}^{-1} {\bf V},
\end{align*}
where ${\bf \Lambda}$ is the eigen-matrix of $ \snr \cdot  \E[\X\X^T] +\I $,  which is a diagonal matrix with the following values along the diagonal: $\lambda_{i}= \snr \sigma_i +1$, and $\sigma_i$ is the $i$-th eigenvalue of  matrix $\E[\X\X^T] $.
Therefore,
\begin{align*}
\Trc \left( \J^2 (\Y) \right) &\ge  \Trc \left(  {\bf V}^{-1} { \bf \Lambda}^{-1} {\bf V} \left( {\bf V}^{-1} { \bf \Lambda}^{-1} {\bf V} \right)^T \right) \notag\\
&= \Trc({ \bf \Lambda}^{-2}) \notag\\
&= \sum_{i=1}^n \frac{1}{ (1+\snr \sigma_i)^2}\\
& \ge  \frac{n}{ (1+\snr )^2},
\end{align*}
where the last inequality comes from minimizing $\sum_{i=1}^n \frac{1}{ (1+\snr \sigma_i)^2}$ subject  to the constraint that $ \Trc \left( \E[\X\X^T] \right) =\sum_{i=1}^n \sigma_i  \le n$ and where the minimum is attained with $\sigma_i=1$  for all $i$.

Finally, note that all inequalities are equalities if $\Y \sim \mathcal{N}({\bf 0}, (1+\snr)\I)$ or  equivalently if $\X \sim \mathcal{N}({\bf 0}, \I)$. This concludes the proof.

\section{Proof of Proposition~\ref{prop: 1/(1+snr) type bound}}
\label{app:prop: 1/(1+snr) type bound}
First observe that  since the conditional expectation is the best estimator under a squared cost function 
\begin{align}
\cov(\X| \Y={\bf y}) &=\E \left[ (\X -\E[\X| \Y]) ( \X -\E[\X| \Y])^{\text{T}} | \Y={\bf y}\right]  \notag\\
& \preceq  \E \left[ (\X -f(\Y)) ( \X -f(\Y))^{\text{T}} | \Y={\bf y}\right], \label{eq: upper bound conditional}
\end{align} 
for any deterministic function $f( \cdot)$. Therefore, the first bound in \eqref{eq: fourth moment bound} follows by choosing $f(\Y)=\frac{\sqrt{\snr} \Y}{1+\snr}$ in \eqref{eq: upper bound conditional}
\begin{align*}
 \Trc \left( \E \left[ \cov^2(\X | \Y) \right]\right)
&  \le   \Trc \left( \E \left[ \E^2 \left[ \left(\X - \frac{\sqrt{\snr} \Y}{1+\snr} \right)  \left( \X -\frac{\sqrt{\snr} \Y}{1+\snr} \right)^{\text{T}} | \Y\right]  \right] \right)\\
& = \frac{1}{(1+\snr)^4} \Trc \left( \E \left[ \E^2 \left[ \left(\X - \sqrt{\snr}\Z \right)  \left( \X -\sqrt{\snr}\Z \right)^{\text{T}} | \Y\right]  \right] \right)\\
 &  \le  \frac{1}{(1+\snr)^4} \Trc \left(  \E \left[ \left( \left(\X - \sqrt{\snr}\Z \right)  \left( \X -\sqrt{\snr}\Z \right)^{\text{T}} \right)^2 \right] \right),
\end{align*}
where the last inequality is due to Jensen's inequality.

The second bound in \eqref{eq: fourth moment bound} follows by choosing  $f(\Y)=0$  in \eqref{eq: upper bound conditional}
\begin{align*}
 \Trc \left( \E \left[ \cov^2(\X | \Y) \right]\right)
&  \le   \Trc \left( \E \left[ \E^2 \left[  (\X -{\bf 0})   (\X  -{\bf 0 })^{\text{T}}| \Y\right]  \right] \right)=\Trc \left( \E \left[ \E^2 \left[  \X   \X  ^{\text{T}}| \Y\right]  \right] \right).
\end{align*}
This concludes the proof. 


\bibliography{refs}
\bibliographystyle{IEEEtran}
\end{document}